\documentclass [onecolumn,11pt]{IEEEtran}
\usepackage{amsmath}
\usepackage{amsfonts}
\usepackage{amssymb}
\usepackage{graphicx}
\usepackage{subfigure}
\usepackage{cite}

\usepackage{doublespace}
\newcommand{\deft}{{\stackrel{\triangle}{=}}}
\newcommand{\eg}{{\em e.g., }}
\newcommand{\ie}{{\em i.e., }}

\newcommand{\diag}{\operatorname{diag} \,}

\newcommand{\tr}{\operatorname{Tr}}

\newcommand{\mvec}{\operatorname{vec}}

\newcommand{\bbi}{{{\bf I}}}

\newcommand{\bbB}{{{\bf B}}}

\newcommand{\bbD}{{{\bf D}}}
\newcommand{\bbd}{{{\bf D}}}

\newcommand{\bbb}{{{\bf B}}}

\newcommand{\bbw}{{{\bf W}}}
\newcommand{\bbx}{{{\bf X}}}
\newcommand{\bby}{{{\bf Y}}}

\newcommand{\bh}{{{\bf h}}}

\newcommand{\bbX}{{{\bf X}}}

\newcommand{\bbm}{{{\bf M}}}
\newcommand{\bx}{{\bf x}}
\newcommand{\bs}{{\bf s}}

\newcommand{\by}{{\bf y}}
\newcommand{\ba}{{\bf a}}

\newcommand{\bo}{{\bf 0}}
\newcommand{\bd}{{\bf d}}
\newcommand{\bc}{{\bf c}}
\newcommand{\bz}{{\bf z}}

\newcommand{\A}{{\mathcal{A}}}
\newcommand{\HH}{{\mathcal{H}}}
\newcommand{\I}{{\mathcal{I}}}

\newcommand{\SSS}{{\mathcal{S}}}

\newcommand{\V}{{\mathcal{V}}}
\newcommand{\U}{{\mathcal{U}}}
\newcommand{\T}{{\mathcal{T}}}

\newcommand{\RR}{{\mathbb{R}}}

\newcommand{\mixed}{{2,\I}}
\newcommand{\st}{\operatorname{s.t.} \,}
\newcommand{\inner}[2]{{\langle#1,#2\rangle}}

\newcommand{\prob}{\operatorname{Prob}}
\newcommand{\sigmamax}{\sigma_\textrm{max}}
\newcommand{\sigmamin}{\sigma_\textrm{min}}
\newcommand{\expp}[1]{e^{#1}}

\newcommand{\bl}{\left(}
\newcommand{\br}{\right)}

\newcommand{\bba}{{\mathbf A}}

\newcommand{\bbs}{{\mathbf S}}

\newcommand{\bbz}{{\mathbf Z}}

\newcommand{\bv}{{\mathbf v}}

\newtheorem{theorem}{Theorem}

\newtheorem{proposition}{Proposition}
\newtheorem{corollary}[theorem]{Corollary}
\newtheorem{definition}{Definition}

\title{\singlespace Robust Recovery of Signals From a Structured Union of
Subspaces}
\author{Yonina C.~Eldar,~\IEEEmembership{Senior~Member,~IEEE} and Moshe Mishali,~\IEEEmembership{Student~Member,~IEEE}\thanks{Department of Electrical Engineering,
Technion---Israel Institute of Technology, Haifa 32000, Israel.
Phone: +972-4-8293256, fax: +972-4-8295757, E-mail:
\{yonina@ee,moshiko@techunix\}.technion.ac.il. This work was
supported in part by the Israel Science Foundation under Grant no.
1081/07 and by the European Commission in the framework of the FP7
Network of Excellence in Wireless COMmunications NEWCOM++
(contract no. 216715).}}
\date{\today}

\begin{document}

\maketitle

\begin{abstract}

Traditional sampling theories consider the problem of
reconstructing an unknown signal $x$ from a series of samples. A
prevalent assumption which often guarantees recovery from the
given measurements is that $x$ lies in a known subspace. Recently,
there has been growing interest in nonlinear but structured signal
models, in which $x$ lies in a union of subspaces. In this paper
we develop a general framework for robust and efficient recovery
of such signals from a given set of samples. More specifically, we
treat the case in which $x$ lies in a sum of $k$ subspaces, chosen
from a larger set of $m$ possibilities. The samples are modelled
as inner products with an arbitrary set of sampling functions. To
derive an efficient and robust recovery algorithm, we show that
our problem can be formulated as that of recovering a block-sparse
vector whose non-zero elements appear in fixed blocks. We then
propose a mixed $\ell_2/\ell_1$ program for block sparse recovery.
Our main result is an equivalence condition under which the
proposed convex algorithm is guaranteed to recover the original
signal. This result relies on the notion of block restricted
isometry property (RIP), which is a generalization of the standard
RIP used extensively in the context of compressed sensing. Based
on RIP we also prove stability of our approach in the presence of
noise and modelling errors.
 A special case of our framework is
that of recovering multiple measurement vectors (MMV) that share a
joint sparsity pattern. Adapting our results to this context leads
to new MMV recovery methods as well as equivalence conditions
under which the entire set can be determined efficiently.

\end{abstract}
\section{Introduction}
\label{sec:intro}

Sampling theory has a rich history dating back to Cauchy.
Undoubtedly, the sampling theorem that had the most impact on
signal processing and communications is that associated with
Whittaker, Kotel\'{n}ikov, and Shannon \cite{S49,J77}. Their famous
result is that a bandlimited function $x(t)$ can be recovered from
its uniform samples as long as the sampling rate exceeds the
Nyquist rate, corresponding to twice the highest frequency of the
signal \cite{N28}. More recently, this basic theorem has been
extended to include more general classes of signal spaces. In
particular, it can be shown that under mild technical conditions,
a signal $x$ lying in a given subspace  can be recovered exactly
from its linear generalized samples using a series of filtering
operations \cite{U00,V01,EM08,ED04}.

Recently, there has been growing interest in nonlinear signal
models in which the unknown $x$ does not necessarily lie in a
subspace. In order to ensure recovery from the samples, some
underlying structure is needed. A general model that captures many
interesting cases is that in which $x$ lies in a union of
subspaces. In this setting, $x$ resides in one of a set of given
subspaces $\V_i$, however, a priori it is not known in which one.
A special case of this framework is the problem underlying the
field of compressed sensing (CS), in which the goal is to recover
a length $N$ vector $\bx$ from $n<N$ linear measurements, where
$\bx$ has no more than $k$ non-zero elements in some basis
\cite{D06,CRT06}. Many algorithms have been proposed in the
literature in order to recover $\bx$ in a stable and efficient
manner \cite{CDS99,MZ93,CRT06,CT05}. A variety of conditions have
been developed to ensure that these methods recover $\bx$ exactly.
One of the main tools in this context is the restricted isometry
property (RIP) \cite{CRT06,CRT06_2,C08}. In particular, it can be
shown that if the measurement matrix satisfies the RIP then $\bx$
can be recovered by solving an $\ell_1$ minimization algorithm.

Another special case of a union of subspaces is the setting in
which the unknown signal $x=x(t)$ has a multiband structure, so
that its Fourier transform consists of a limited number of bands
at unknown locations \cite{ME07,ME09}. By formulating this problem
within the framework of CS, explicit sub-Nyquist sampling and
reconstruction schemes were developed in \cite{ME07,ME09} that
ensure perfect-recovery at the minimal possible rate. This setup
was recently generalized in \cite{E08,E082} to deal with sampling
and reconstruction of signals that lie in a finite union of
shift-invariant subspaces. By combining ideas from standard
sampling theory with CS results \cite{ME08}, explicit low-rate
sampling and recovery methods were developed for such signal sets.
Another example of a union of subspaces is the set of finite rate
of innovation signals \cite{VMB02,DVB07}, that are modelled as a
weighted sum of shifts of a given generating function, where the
shifts are unknown.

In this paper, our goal is to develop a unified framework for
efficient recovery of signals that lie in a structured union of
subspaces. Our emphasis is on computationally efficient methods
that are stable in the presence of noise and modelling errors. In
contrast to our previous work \cite{E08,E082,ME09,ME07}, here we
consider unions of finite-dimensional subspaces. Specifically, we
restrict our attention to the case in which $x$ resides in a sum
of $k$ subspaces, chosen from a given set of $m$ subspaces $\A_j,1
\leq j \leq m$. However, which subspaces comprise the sum is
unknown. This setting is a special case of the more general union
model considered in \cite{LD08,BD09}. Conditions under which
unique and stable sampling are possible were developed in
\cite{LD08,BD09}. However, no concrete algorithm was provided to
recover such a signal from a given set of samples in a stable and
efficient manner. Here we propose a convex optimization algorithm
that will often recover the true underlying $x$, and develop
explicit conditions under which perfect recovery is guaranteed.
Furthermore, we prove that our method is stable and robust in the
sense that the reconstruction error is bounded in the presence of
noise and mismodelling, namely when $x$ does not lie exactly in
the union. Our results rely on a generalization of the RIP which
fits the union setting we treat here.

Our first contribution is showing that the problem of recovering
$x$ in a structured union of subspaces can be cast as a sparse
recovery problem, in which it is desired to recover a sparse
vector $\bc$ that has a particular sparsity pattern: the non-zero
values appear in fixed blocks. We refer to such a model as block
sparsity. Clearly any block-sparse vector is also sparse in the
standard sense. However, by exploiting the block structure of the
sparsity pattern, recovery may be possible under more general
conditions.

Next, we develop a concrete algorithm to recover a block-sparse
vector from given measurements, which is based on minimizing a
mixed $\ell_2/\ell_1$ norm. This problem can be cast as a convex
second order cone program (SOCP), and solved efficiently using
standard software packages. A mixed norm approach for block-sparse
recovery was also considered in \cite{SPH08,PVMH08}. By analyzing
the measurement operator's null space, it was shown that
asymptotically, as the signal length grows to infinity, and under
ideal conditions (no noise or modeling errors), perfect recovery
is possible with high probability. However, no robust equivalence
results were established between the output of the algorithm and
the true block-sparse vector for a given finite-length measurement
vector, or in the presence of noise and mismodelling.

Generalizing the concept of RIP to our setting, we introduce the
block RIP, which is a less stringent requirement. We then prove
that if the measurement matrix satisfies the block RIP, then our
proposed convex algorithm will recover the underlying block sparse
signal. Furthermore, under block RIP, our algorithm is stable in
the presence of noise and mismodelling errors. Using ideas similar
to \cite{CT05},\cite{CT06} we then prove that random matrices
satisfy the block RIP with overwhelming probability. Moreover, the
probability to satisfy the block RIP is substantially larger than
that of satisfying the standard RIP. These results establish that
a signal $x$ that lies in a finite structured union can be
recovered efficiently and stably with overwhelming probability if
a certain measurement matrix is constructed from a random
ensemble.

An interesting special case of the block-sparse model is the
multiple measurement vector (MMV) problem, in which we have a set
of unknown vectors that share a joint sparsity pattern. MMV
recovery algorithms were studied in
\cite{Cotter,Chen,TroppI,TroppII,ME08}. Equivalence results based
on mutual coherence for a mixed $\ell_p/\ell_1$ program were
derived in \cite{Chen}. These results turn out to be the same as
that obtained from a single measurement problem. This is in
contrast to the fact that in practice, MMV methods tend to
outperform algorithms that treat each of the vectors separately.
In order to develop meaningful equivalence results,  we cast the
MMV problem as one of block-sparse recovery. Our mixed
$\ell_2/\ell_1$ method translates into minimizing the sum of the
$\ell_2$ row-norms of the unknown matrix representing the MMV set.
Our general results lead to RIP-based equivalence conditions for
this algorithm. Furthermore, our framework suggests a different
type of sampling method for MMV problems which tends to increase
the recovery rate. The equivalence condition we obtain in this
case is stronger than the single measurement setting. As we show,
this method leads to superior recovery rate when compared with
other popular MMV algorithms.

The remainder of the paper is organized as follows. In
Section~\ref{sec:union} we describe the general problem of
sampling from a union of subspaces. The relationship between our
problem and that of block-sparse recovery is developed in
Section~\ref{sec:block}. In Section~\ref{sec:unique} we explore
stability and uniqueness issues which leads to the definition of
block RIP. We also present a non-convex optimization algorithm
with combinatorial complexity whose solution is the true unknown
$x$. A convex relaxation of this algorithm is proposed in
Section~\ref{sec:l1}. We then derive equivalence conditions based
on block RIP. The concept of block RIP is further used to
establish robustness and stability of our algorithm in the
presence of noise and modelling errors. This approach is
specialized to MMV sampling in Section~\ref{sec:mmv}. Finally, in
Section~\ref{sec:prob} we prove that random ensembles tend to
satisfy the block RIP with high probability.

 Throughout the paper, we denote
vectors in an arbitrary Hilbert space $\HH$ by lower case letters
\eg $x$, and sets of vectors in $\HH$ by calligraphic letters, \eg
$\SSS$. Vectors in $\RR^N$ are written as boldface lowercase
letters \eg $\bx$, and matrices as boldface uppercase letters \eg
$\bba$. The identity matrix of appropriate dimension is written as
$\bbi$ or $\bbi_d$ when the dimension is not clear from the
context, and $\bba^T$ is the transpose of the matrix $\bba$. The
$i$th element of a vector $\bx$ is denoted by $\bx(i)$. Linear
transformations from $\RR^n$ to $\HH$ are written as upper case
letters $A:\RR^n \rightarrow \HH$. The adjoint of $A$ is written
as $A^*$. The standard Euclidean norm is denoted
$\|\bx\|_2=\sqrt{\bx^T\bx}$ and $\|\bx\|_1=\sum_i |\bx(i)|$ is the
$\ell_1$ norm of $\bx$. The Kronecker product between matrices
$\bba$ and $\bbb$ is denoted $\bba \otimes \bbb$. The following
variables are used in the sequel: $n$ is the number of samples,
$N$ is the length of the input signal $\bx$ when it is a vector,
$k$ is the sparsity or block sparsity (to be defined later on) of
a vector $\bc$, and $m$ is the number of subspaces. For ease of
notation we assume throughout that all scalars are defined over
the field of real numbers; however, the results are also valid
over the complex domain with appropriate modifications.

\section{Union of Subspaces}
\label{sec:union}

\subsection{Subspace Sampling}
Traditional sampling theory deals with the problem of recovering
an unknown signal $x \in \HH$ from a set of $n$ samples
$y_i=f_i(x)$ where $f_i(x)$ is some function of $x$. The signal
$x$ can be a function of time $x=x(t)$, or can represent a
finite-length vector $x=\bx$. The most common type of sampling is
linear sampling in which
\begin{equation}
y_i=\inner{s_i}{x},\quad 1 \leq i \leq n,
\end{equation}
for a set of functions $s_i \in \HH$
\cite{M89,U00,DV97,EW05,CE03,E02,E04s,EC04}.
 Here $\inner{x}{y}$
denotes the standard inner product on $\HH$. For example, if
$\HH=L_2$ is the space of real finite-energy signals then
\begin{equation}
\inner{x}{y}=\int_{ -\infty}^\infty x(t)y(t)dt.
\end{equation}
When $\HH=\RR^N$ for some $N$,
\begin{equation}
\inner{\bx}{\by}=\sum_{i=1}^N\bx(i)\by(i).
\end{equation}
Nonlinear sampling is treated in \cite{DE08}. However, here our
focus will be on the linear case.

When $\HH=\RR^N$ the unknown $x=\bx$ as well as the sampling
functions $s_i=\bs_i$ are vectors in $\RR^N$. Therefore, the
samples can be written conveniently in matrix form as
$\by=\bbs^T\bx$, where $\bbs$ is the matrix with columns $\bs_i$.
In the more general case in which $\HH=L_2$ or any other abstract
Hilbert space, we can use the set transformation notation in order
to conveniently represent the samples. A set transformation
$S:\RR^n\rightarrow\HH$ corresponding to sampling vectors $\{s_i
\in \HH,1 \leq i \leq n\}$ is defined by
\begin{equation}
S \bc=\sum_{i=1}^n \bc(i) s_i
\end{equation}
 for all $\bc\in \RR^n$. From the
definition of the adjoint, if $\bc=S^*x$, then
$\bc(i)=\inner{s_i}{x}$. Note that when $\HH=\RR^N$, $S=\bbs$ and
$S^*=\bbs^T$. Using this notation, we can always express the
samples as
\begin{equation}
\label{eq:samples} \by=S^*x,
\end{equation}
where $S$ is a set transformation for arbitrary $\HH$, and an
appropriate matrix when $\HH=\RR^N$.

 Our goal is to
recover $x$ from the samples $\by \in \RR^n$. If the vectors $s_i$
do not span the entire space $\HH$, then there are many possible
signals $x$ consistent with $\by$. More specifically, if we define
by $\SSS$ the sampling space spanned by the vectors $s_i$, then
clearly $S^*v=0$ for any $v \in \SSS^\perp$. Therefore, if
$\SSS^\perp$ is not the trivial space then adding such a vector
$v$ to any solution $x$ of (\ref{eq:samples}) will result in the
same samples $\by$.
 However, by
exploiting prior knowledge on $x$, in many cases uniqueness can be
guaranteed. A prior very often assumed is that $x$ lies in a given
subspace $\A$ of $\HH$ \cite{U00,V01,EM08,ED04}. If $\A$ and
$\SSS$ have the same finite dimension, and $\SSS^\perp$ and $\A$
intersect only at the $0$ vector, then  $x$ can be perfectly
recovered from the samples $\by$ \cite{ED04,EM08,ME09b}.

\subsection{Union of Subspaces}

When subspace information is available, perfect reconstruction can
often be guaranteed. Furthermore, recovery can be implemented by a
simple linear transformation of the given samples
(\ref{eq:samples}). However, there are many practical scenarios in
which we are given prior information about $x$ that is not
necessarily in the from of a subspace. One such case studied in
detail in \cite{ME09b} is that in which $x$ is known to be smooth.
Here we focus our attention on the setting where $x$ lies in a
union of subspaces
\begin{equation}
\label{eq:union} \U=\bigcup_i \V_i
\end{equation}
 where each $\V_i$ is a subspace. Thus, $x$
belongs to one of the $\V_i$, but we do not know a priori to which
one \cite{LD08,BD09}. Note that the set $\U$ is no longer a
subspace. Indeed, if $\V_i$ is, for example, a one-dimensional
space spanned by the vector $\bv_i$, then $\U$ contains vectors of
the form $\alpha \bv_i$ for some $i$ but does not include their
linear combinations. Our goal is to recover a vector $x$ lying in
a union of subspaces, from a given set of samples. In principle,
if we knew which subspace $x$ belonged to, then reconstruction can
be obtained using standard sampling results. However, here the
problem is more involved because conceptually we first need to
identify the correct subspace and only then can we recover the
signal within the space.

Previous work on sampling over a union focused on invertibility
and stability results \cite{LD08,BD09}. In contrast, here, our
main interest is in developing concrete recovery algorithms that
are provably robust. To achieve this goal, we limit our attention
to a subclass of (\ref{eq:union}) for which stable recovery
algorithms can be developed and analyzed. Specifically, we treat
the case in which each $\V_i$ has the additional structure
\begin{equation}
\label{eq:unionv} \V_i=\bigoplus_{|j|=k} \A_j,
\end{equation}
where $\{\A_j,1 \leq j \leq m\}$ are a given set of disjoint
subspaces, and $|j|=k$ denotes a sum over $k$ indices. Thus, each
subspace $\V_i$ corresponds to a different choice of $k$ subspaces
$\A_j$ that comprise the sum. We assume throughout the paper that
$m$ and
 the dimensions $d_i=\dim(\A_i)$ of the subspaces $\A_i$
are finite. Given $n$ samples
\begin{equation}
\by=S^*x
\end{equation}
and the knowledge that $x$ lies in exactly one of the subspaces
$\V_i$, we would like to recover the unknown signal $x$. In this
setting, there are $\binom{m}{k}$ possible subspaces comprising
the union.

An alternative interpretation of our model is as follows. Given an
observation vector $\by$, we seek a signal $x$ for which
$\by=S^*x$ and in addition $x$ can be written as
\begin{equation}
x=\sum_{i=1}^k x_i,
\end{equation}
where each $x_i$ lies in $\A_j$ for some index $j$.

 A special case is the standard CS problem in
which $x=\bx$ is a vector of length $N$, that has a sparse
representation in a given basis defined by an invertible matrix
$\bbw$. Thus, $\bx=\bbw \bc$ where $\bc$ is a sparse vector that
has at most $k$ nonzero elements. This fits our framework by
choosing $\A_i$ as the space spanned by the $i$th column of
$\bbw$. In this setting $m=N$, and there are $\binom{N}{k}$
subspaces comprising the union.

Another example is the block sparsity model \cite{SPH08,EB09} in
which $\bx$ is divided into equal-length blocks of size $d$, and
at most $k$ blocks can be non zero. Such a vector can be described
in our setting with $\HH=\RR^N$ by choosing $\A_i$ to be the space
spanned by the corresponding $i$ columns of the identity matrix.
Here $m=N/d$ and there are $\binom{N/d}{k}$ subspaces in the
union.

A final example is the MMV problem
\cite{Cotter,Chen,TroppI,TroppII,ME08} in which our goal is to
recover a matrix $\bbx$ from measurements $\bby=\bbm \bbx$, for a
given sampling matrix $\bbm$. The matrix $\bbx$ is assumed to have
at most $k$ non-zero rows. Thus, not only is each column $\bx_i$
$k$-sparse, but in addition the non-zero elements of $\bx_i$ share
a joint sparsity pattern. This problem can be transformed into
that of recovering a $k$-block sparse signal by stacking the rows
of $\bbx$ and $\bby$, leading to the relationship
\begin{equation}
\label{eq:mmvb} \mvec(\bby^T)=(\bbm \otimes \bbi)\mvec(\bbx^T).
\end{equation}
The structure of $\bbx$ leads to a vector $\mvec(\bbx^T)$ that is
$k$-block sparse.

\subsection{Problem Formulation and Main Results}

Given $k$ and the subspaces $\A_i$, we would like to address the
following questions:
\begin{enumerate}
\item What are the conditions on the sampling vectors $s_i,1 \leq
i \leq n$ in order to guarantee that the sampling is invertible
and stable? \item How can we recover the unique $x$ (regardless of
computational complexity)?\item How can we recover the unique $x$
in an efficient and stable manner?
\end{enumerate}
The first question was addressed in \cite{LD08,BD09} in the more
general context of unions of spaces (without requiring a
particular structure such as (\ref{eq:unionv})). However, no
concrete methods were proposed in order to recover $x$. Here we
provide efficient convex algorithms that recover $x$ in a stable
way for arbitrary $k$ under appropriate conditions on the sampling
functions $s_i$ and the spaces $\A_i$.

Our results are based on an equivalence between the union of
subspaces problem assuming (\ref{eq:unionv}) and that of
recovering block-sparse vectors. This allows us to recover $x$
from the given samples by first treating the problem of recovering
a block $k$-sparse vector $\bc$ from a given set of measurements.
This relationship is established in the next section. In the
reminder of the paper we therefore focus on the block $k$-sparse
model and develop our results in that context. In particular, we
introduce a block RIP condition that ensures uniqueness and
stability of our sampling problem. We then suggest an efficient
convex optimization problem which approximates an unknown
block-sparse vector $\bc$. Based on block RIP we prove that $\bc$
can be recovered exactly in a stable way using the proposed
optimization program. Furthermore, in the presence of noise and
modeling errors, our algorithm can approximate the best block-$k$
sparse solution.

\section{Connection with Block Sparsity}\label{sec:block}

 Consider the model of a signal $x$ in the union of $k$
out of $m$ subspaces $\A_i$, with $d_i=\dim(\A_i)$ as in
(\ref{eq:union}) and (\ref{eq:unionv}). To write $x$ explicitly,
we choose a basis for each $\A_i$. Denoting by $A_i:\RR^{d_i}
\rightarrow \HH$ the set transformation corresponding to a basis
for $\A_i$, any such $x$ can be written as
\begin{equation}
\label{eq:models} x=\sum_{|i|=k} A_i \bc_i,
\end{equation}
where $\bc_i \in \RR^{d_i}$ are the representation coefficients in
$\A_i$, and  $|i|=k$ denotes a sum over a set of $k$ indices. The
choice of indices depend on the signal $x$ and are unknown in
advance.

To develop the equivalence with block sparsity, it is useful to
introduce some further notation. First, we define $A:\RR^N
\rightarrow \HH$ as the set transformation that is a result of
concatenating the different $A_i$, with
\begin{equation}
\label{eq:N} N=\sum_{i=1}^m d_i.
\end{equation}
Next, we define the $i$th sub-block $\bc[i]$ of a length-$N$
vector $\bc$ over $\I=\{d_1,\ldots,d_m\}$. The $i$th sub-block is
of length $d_i$, and the blocks are formed sequentially so that
\begin{equation}
\label{eq:xblock} \bc^T=[\underbrace{c_1 \,\, \ldots \,\,
c_{d_1}}_{\bc[1]} \,\, \ldots\,\,
\underbrace{c_{N-d_m+1}\,\,\ldots \,\,c_{N}}_{\bc[m]}]^T.
\end{equation}
We can then define $A$ by
\begin{equation}
\label{eq:A} A \bc=\sum_{i=1}^m A_i \bc[i].
\end{equation}
When $\HH=\RR^N$ for some $N$, $A_i=\bba_i$ is a matrix and
$A=\bba$ is the matrix obtained by column-wise concatenating
$\bba_i$. If for a given $x$ the $j$th subspace $\A_j$ does not
appear in the sum (\ref{eq:unionv}), or equivalently in
(\ref{eq:models}), then $\bc[j]=\bo$.

Any $x$ in the union (\ref{eq:union}), (\ref{eq:unionv}) can be
represented in terms of $k$ of the bases $A_i$. Therefore, we can
write $x=A\bc$ where there are at most $k$ non-zero blocks
$\bc[i]$. Consequently, our union model is equivalent to the model
in which $x$ is represented by a sparse vector $\bc$ in an
appropriate basis. However, the sparsity pattern here has a unique
form which we will exploit in our conditions and algorithms: the
non-zero elements appear in blocks.
\begin{definition} A
vector $\bc \in \RR^N$ is called block $k$-sparse over
$\I=\{d_1,\ldots,d_m\}$ if $\bc[i]$ is nonzero for at most $k$
indices $i$ where $N=\sum_i d_i$.
\end{definition}
An example of a block-sparse vector with $k=2$ is depicted in
Fig.~\ref{FigBlockSparsity}.
\begin{figure}[h]
\centering
\includegraphics[scale=1]{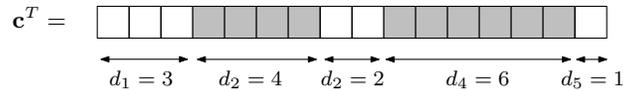}
\caption{A block-sparse vector $\bc$ over
$\I=\{d_1,\dots,d_5\}$. The gray areas represent
10 non-zero entries which occupy two
blocks.}\label{FigBlockSparsity}
\end{figure}
 When $d_i=1$ for each $i$, block
sparsity reduces to the conventional definition of a sparse
vector. Denoting
\begin{equation}
\|\bc\|_{0,\I} = \sum_{i=1}^m I(\|\bc[i]\|_2>0),
\end{equation}
where $I(\|\bc[i]\|_2>0)$ is an indicator function that obtains
the value $1$ if $\|\bc[i]\|_2>0$ and $0$ otherwise, a block
$k$-sparse vector $\bc$ can be defined by $\|\bc\|_{0,\I} \leq k$.

Evidently, there is a one-to-one correspondence between a vector
$x$ in the union, and a block-sparse vector $\bc$. The
measurements (\ref{eq:samples}) can also be represented explicitly
in terms of $\bc$ as
\begin{equation}
\label{eq:samplesd} \by=S^*x=S^*A\bc=\bbd\bc,
\end{equation}
where $\bbd$ is the $n \times N$ matrix defined by
\begin{equation}
\label{eq:D} \bbd=S^*A.
\end{equation}
We can therefore phrase our problem in terms of $\bbd$ and $\bc$
as that of recovering a block-$k$ sparse vector $\bc$ over $\I$
from the measurements (\ref{eq:samplesd}).

Note that the choice of basis $A_i$ for each subspace does not
affect our model. Indeed, choosing alternative bases will lead to
$x=A\bbw \bc$ where $\bbw$ is a block diagonal matrix with blocks
of size $d_i$. Defining $\tilde{\bc}=\bbw\bc$, the block sparsity
pattern of $\tilde{\bc}$ is equal to that of $\bc$.

Since our problem is equivalent to that of recovering a block
sparse vector over $\I$ from linear measurements $\by=\bbd\bc$, in
the reminder of the paper we focus our attention on this problem.

\section{Uniqueness and Stability}
\label{sec:unique}

In this section we study the uniqueness and stability of our
sampling method. These properties are intimately related to the
RIP, which we generalize here to the block-sparse setting.

 The first question we
address is that of uniqueness, namely conditions under which a
block-sparse vector $\bc$ is uniquely determined by the
measurement vector $\by=\bbd\bc$.
\begin{proposition}
\label{prop:inv} There is a unique block-$k$ sparse vector $\bc$
consistent with the measurements $\by=\bbd\bc$ if and only if
$\bbd \bc \neq 0$ for every $\bc \neq 0$ that is block
$2k$-sparse.
\end{proposition}
\begin{proof}
The proof follows from \cite[Proposition 4]{LD08}.
\end{proof}

We next address the issue of stability. A sampling operator is
stable for a set $\T$ if and only if there exists constants
$\alpha>0$, $\beta<\infty$ such that
\begin{equation}\label{eq:stablity}
\alpha \|x_1-x_2\|^2_\HH \leq  \|S^*x_1-S^*x_2\|^2_2 \leq  \beta
\|x_1-x_2\|^2_\HH,
\end{equation}
for every $x_1,x_2$ in $\T$. The ratio $\kappa=\beta/\alpha$
provides a measure for stability of the sampling operator. The
operator is maximally stable when $\kappa=1$. In our setting,
$S^*$ is  replaced by $\bbd$, and the set $\T$ contains block-$k$
sparse vectors. The following proposition follows immediately from
(\ref{eq:stablity}) by noting that given two block-$k$ sparse
vectors $\bc_1,\bc_2$ their difference $\bc_1-\bc_2$ is block-$2k$
sparse.
\begin{proposition}
\label{prop:stab} The measurement matrix $\bbd$ is stable for
every block $k$-sparse vector $\bc$ if and only if there exists
$C_1>0$ and $C_2<\infty$ such that
\begin{equation}
\label{eq:stab} C_1 \|\bv\|^2_2 \leq \|\bbd\bv\|^2_2 \leq C_2
\|\bv\|^2_2,
\end{equation}
for every $\bv$ that is block $2k$-sparse.
\end{proposition}
It is easy to see that if $\bbd$ satisfies (\ref{eq:stab}) then
$\bbd\bc \neq \bo$ for all block $2k$-sparse vectors $\bc$.
Therefore, this condition implies both invertibility and
stability.

\subsection{Block RIP}

Property (\ref{eq:stab}) is related to the RIP used in several
previous works in CS \cite{CRT06,CRT06_2,C08}. A matrix $\bbd$ of
size $n \times N$ is said to have the RIP if there exists a
constant $\delta_k\in[0,1)$ such that for every $k$-sparse $\bc
\in \RR^N$,
\begin{equation}
\label{eq:ripo} (1-\delta_k)\|\bc\|_2^2 \leq \|\bbd\bc\|_2^2 \leq
(1+\delta_k) \|\bc\|_2^2.
\end{equation}
Extending this property to block-sparse vectors leads to the
following definition:
\begin{definition}
Let $\bbd:\RR^N \rightarrow \RR^n$ be a given matrix. Then $\bbd$
has the block RIP over $\I=\{d_1,\ldots,d_m\}$ with
parameter $\delta_{k|\I}$ if for every $\bc \in \RR^N$ that is block
$k$-sparse over $\I$ we have that
\begin{equation}
\label{eq:rip} (1-\delta_{k|\I})\|\bc\|_2^2 \leq \|\bbd\bc\|_2^2
\leq (1+\delta_{k|\I}) \|\bc\|_2^2.
\end{equation}
\end{definition}
By abuse of notation, we use $\delta_k$ for the block-RIP constant
$\delta_{k|\I}$ when it is clear from the context that we refer to
blocks. Block-RIP is a special case of the $\A$-restricted
isometry defined in \cite{BD09}. From Proposition~\ref{prop:inv}
it follows that if $\bbd$ satisfies the RIP (\ref{eq:rip}) with
$\delta_{2k}<1$, then there is a unique block-sparse vector $\bc$
consistent with (\ref{eq:samplesd}).

Note that a block $k$-sparse vector over $\I$ is $M$-sparse in the
conventional sense where $M$ is the sum of the $k$ largest values
in $\I$, since it has at most $M$ nonzero elements. If we require
$\bbd$ to satisfy RIP for all $M$-sparse vectors, then
(\ref{eq:rip}) must hold for all $2M$-sparse vectors $\bc$. Since
we only require the RIP for block sparse signals, (\ref{eq:rip})
only has to be satisfied for a certain subset of $2M$-sparse
signals, namely those that have block sparsity. As a result, the
block-RIP constant $\delta_{k|\I}$ is typically smaller than
$\delta_{M}$ (where $M$ depends on $k$; for blocks with equal size
$d$, $M=kd$).

To emphasize the advantage of block RIP over standard RIP,
consider the following matrix, separated into three blocks of two
columns each:
\begin{equation}\label{matD}
\bbD = \left(
         \begin{array}{cc|cc|cc}
           -1 & 1 &  0 &  0 &  0 & 1 \\
            0 & 2 & -1 &  0 &  0 & 3 \\
            0 & 3 &  0 & -1 &  0 & 1 \\
            0 & 1 &  0 &  0 & -1 & 1 \\
         \end{array}
       \right)\cdot \bbB,
\end{equation}
where $\bbB$ is a diagonal matrix that results in unit-norm
columns of $\bbD$, \ie $\bbB=\diag(1,15,1,1,1,12)^{-1/2}$.
 In
this example $m=3$ and $\I=\{d_1=2,d_2=2,d_3=2\}$. Suppose that
$\bc$ is block-1 sparse, which corresponds to at most two non-zero
values. Brute-force calculations show that the smallest value of
$\delta_2$ satisfying the standard RIP (\ref{eq:ripo}) is
$\delta_2 = 0.866$. On the other hand, the block-RIP
(\ref{eq:rip}) corresponding to the case in which the two non-zero
elements are restricted to occur in one block is satisfied with
$\delta_{1|\I} = 0.289$. Increasing the number of non-zero
elements to $k=4$, we can verify that the standard RIP
(\ref{eq:ripo}) does not hold for any $\delta_4\in[0,1)$. Indeed,
in this example there exist two $4$-sparse vectors that result in
the same measurements. In contrast, $\delta_{2|\I}=0.966$
satisfies the lower bound in (\ref{eq:rip}) when restricting the
$4$ non-zero values to  two blocks. Consequently, the measurements
$\by = \bbD \bc$ uniquely specify a single block-sparse $\bc$.

In the next section, we will see that the ability to recover $\bc$
in a computationally efficient way depends on the constant
$\delta_{2k|\I}$ in the block RIP (\ref{eq:rip}). The smaller the
value of $\delta_{2k|\I}$, the fewer samples are needed in order
to guarantee stable recovery. Both standard and block RIP
constants $\delta_k,\delta_{k|\I}$ are by definition increasing
with $k$. Therefore, it was suggested in \cite{CT05} to normalize
each of the columns of $\bbD$ to 1, so as to start with
$\delta_1=0$. In the same spirit, we recommend choosing the bases
for $\A_i$ such that $\bbd=S^*A$ has unit-norm columns,
corresponding to $\delta_{1|I}=0$.

\subsection{Recovery Method}

We have seen that if $\bbd$ satisfies the RIP (\ref{eq:rip}) with
$\delta_{2k}<1$, then there is a unique block-sparse vector $\bc$
consistent with (\ref{eq:samplesd}). The question is how to find
$\bc$ in practice. Below we present an algorithm that will in
principle find the unique $\bc$ from the samples $\by$.
Unfortunately, though, it has exponential complexity. In the next
section we show that under a stronger condition on $\delta_{2k}$
we can recover $\bc$ in a stable and efficient manner.

 Our
first claim is that $\bc$ can be uniquely recovered by solving the
optimization problem
\begin{eqnarray}
\label{eq:l0}
\min_{\bc} && \|\bc\|_{0,\I} \nonumber \\
\st && \by=\bbd\bc.
\end{eqnarray}
 To show that (\ref{eq:l0}) will
indeed recover the true value of $\bc$, suppose that there exists
a $\bc'$ such that $\bbd\bc'=\by$ and $ \|\bc'\|_{0,\I} \leq
\|\bc\|_{0,\I} \leq k$. Since both $\bc,\bc'$ are consistent with
the measurements,
\begin{equation}
\bo=\bbd(\bc-\bc')=\bbd\bd,
\end{equation}
where $\|\bd\|_{0,\I} \leq 2k$ so that $\bd$ is a block
$2k$-sparse vector. Since $\bbd$ satisfies (\ref{eq:rip}) with
$\delta_{2k}<1$, we must have that $\bd=\bo$ or $\bc=\bc'$.

In principle (\ref{eq:l0}) can be solved by searching over all
possible sets of $k$ blocks whether there exists a $\bc$ that is
consistent with the measurements. The invertibility condition
(\ref{eq:rip}) ensures that there is only one such $\bc$. However,
clearly this approach is not efficient.

\section{Convex Recovery Algorithm}\label{sec:l1}

\subsection{Noise-Free Recovery}

We now develop an efficient convex optimization problem instead of
(\ref{eq:l0}) to approximate $\bc$. As we show, if $\bbd$
satisfies (\ref{eq:rip}) with a small enough value of
$\delta_{2k}$, then the method we propose will recover $\bc$
exactly.

Our approach is to minimize the sum of the energy of the blocks
$\bc[i]$. To write down the problem explicitly, we define the
mixed $\ell_2/\ell_1$ norm over the index set
$\I=\{d_1,\ldots,d_m\}$ as
\begin{equation}
\|\bc\|_\mixed = \sum_{i=1}^m \|\bc[i] \|_2.
\end{equation}
The algorithm we suggest is then
\begin{eqnarray}\label{mainalg}
\label{eq:l1}
\min_{\bc} && \|\bc\|_{2,\I} \nonumber \\
\st && \by=\bbd\bc.
\end{eqnarray}
Problem (\ref{eq:l1}) can be written as an SOCP by defining
$t_i=\|\bc[i]\|_2$. Then (\ref{eq:l1}) is equivalent to
\begin{eqnarray}
\label{eq:socp}
\min_{\bc,t_i} && \sum_{i=1}^m t_i \nonumber \\
\st && \by=\bbd\bc \nonumber \\
&& t_i \geq \|\bc[i]\|_2,\quad 1 \leq i \leq m \nonumber \\
&& t_i \geq 0,\quad 1 \leq i \leq m,
\end{eqnarray}
which can be solved using standard software packages.

The next theorem establishes that the solution to (\ref{eq:l1}) is
the true $\bc$ as long as $\delta_{2k}$ is small enough.
\begin{theorem}
\label{thm:rip} Let $\by=\bbd \bc_0$ be measurements of a block
$k$-sparse vector $\bc_0$. If $\bbd$ satisfies the block RIP
(\ref{eq:rip}) with $\delta_{2k}<\sqrt{2}-1$ then
\begin{enumerate}
\item there is a unique block-$k$ sparse vector $\bc$ consistent
with $\by$; \item the SOCP (\ref{eq:socp}) has a unique solution;
\item the solution to the SOCP is equal to $\bc_0$.
\end{enumerate}
\end{theorem}
Before proving the theorem we note that it provides a gain over
standard CS results. Specifically, it is shown in \cite{C08} that
if $\bc$ is $k$-sparse and the measurement matrix $\bbd$ satisfies
the standard RIP with $\delta_{2k}<\sqrt{2}-1$, then $\bc$ can be
recovered exactly from the measurements $\by=\bbd\bc$ via the
linear program:
\begin{eqnarray}
\label{eq:l1cs}
\min_{\bc} && \|\bc\|_1 \nonumber \\
\st && \by=\bbd\bc.
\end{eqnarray}
Since any block $k$-sparse vector is also $M$-sparse with $M$
equal to the sum of the $k$ largest values of $d_i$, we can find
$\bc_0$ of Theorem~\ref{thm:rip} by solving (\ref{eq:l1cs}) if
$\delta_{2M}$ is small enough.  However, this standard CS approach
does not exploit the fact that the non-zero values appear in
blocks, and not in arbitrary locations within the vector $\bc_0$.
On the other hand, the SOCP (\ref{eq:socp}) explicitly takes the
block structure of $\bc_0$ into account. Therefore, the condition
of Theorem~\ref{thm:rip} is not as stringent as that obtained by
using equivalence results with respect to (\ref{eq:l1cs}). Indeed,
the block RIP (\ref{eq:rip}) bounds the norm of $\|\bbd\bc\|$ over
block sparse vectors $\bc$, while the standard RIP considers all
possible choices of $\bc$, also those that are not $2k$-block
sparse. Therefore, the value of $\delta_{2k}$ in (\ref{eq:rip})
can be lower than that obtained from (\ref{eq:ripo}) with $k=2M$,
as we illustrated by an example in Section~\ref{sec:block}. This
advantage will also be seen in the context of a concrete example
at the end of the section.

Our proof below is rooted in that of \cite{C08}. However, some
essential modifications are necessary in order to adapt the
results to the block-sparse case.  These differences are a result
of the fact that our algorithm relies on the mixed $\ell_2/\ell_1$
norm rather than the $\ell_1$ norm alone. This adds another layer
of complication to the proof, and therefore we expand the
derivations in more detail than in \cite{C08}.

\begin{proof}
 We first
note that $\delta_{2k}<1$ guarantees uniqueness of $\bc_0$ from
Proposition~\ref{prop:inv}. To prove parts 2) and 3) we show that
any solution to (\ref{eq:l1}) has to be equal to $\bc_0$. To this
end let $\bc'=\bc_0+\bh$ be a solution of (\ref{eq:l1}). The true
value $\bc_0$ is non-zero over at most $k$ blocks. We denote by
$\I_0$ the block indices for which $\bc_0$ is nonzero, and by
$\bh_{\I_0}$ the restriction of $\bh$ to these blocks. Next we
decompose $\bh$ as
\begin{equation}
\bh=\sum_{i=0}^{\ell-1} \bh_{\I_i},
\end{equation}
where $\bh_{\I_i}$ is the restriction of $\bh$ to the set $\I_i$
which consists of $k$ blocks, chosen such that the norm of
$\bh_{\I_0^c}$ over $\I_1$  is largest, the norm over $\I_2$ is
second largest and so on. Our goal is to show that $\bh=\bo$. We
prove this by noting that
\begin{equation}
\label{eq:normh} \|\bh\|_2 =\|\bh_{\I_0 \cup \I_1}+\bh_{(\I_0 \cup
\I_1)^c}\|_2\leq \|\bh_{\I_0 \cup \I_1}\|_2+\|\bh_{(\I_0 \cup
\I_1)^c}\|_2.
\end{equation}
In the first part of the proof we show that $\|\bh_{(\I_0 \cup
\I_1)^c}\|_2 \leq \|\bh_{\I_0 \cup \I_1}\|_2$. In the second part
we establish that $\|\bh_{\I_0 \cup \I_1}\|_2=0$, which completes
the proof.

{\it Part I:$\|\bh_{(\I_0 \cup \I_1)^c}\|_2 \leq \|\bh_{\I_0 \cup
\I_1}\|_2$}

We begin by noting that
\begin{equation}
\label{eq:hcn} \|\bh_{(\I_0 \cup
\I_1)^c}\|_2=\left\|\sum_{i=2}^{\ell-1} \bh_{\I_i}\right\|_2 \leq
\sum_{i=2}^{\ell-1}\|\bh_{\I_i}\|_2.
\end{equation}
Therefore, it is sufficient to bound $\|\bh_{\I_i}\|_2$ for $i
\geq 2$. Now,
\begin{equation}
\label{eq:l2n} \|\bh_{\I_i}\|_2 \leq k^{1/2}
\|\bh_{\I_i}\|_{\infty,\I} \leq k^{-1/2}
\|\bh_{\I_{i-1}}\|_\mixed,
\end{equation}
where we defined $\|\ba\|_{\infty,\I}=\max_i \|\ba[i]\|_2$. The
first inequality follows from the fact that for any block
$k$-sparse $\bc$,
\begin{equation}
\|\bc\|_2^2=\sum_{|i|=k} \|\bc[i]\|_2^2 \leq k
\|\bc\|_{\infty,\I}^2.
\end{equation}
The second inequality in (\ref{eq:l2n}) is a result of the fact
that the norm of each block in $\bh_{\I_i}$ is by definition
smaller or equal to the norm of each block in $\bh_{\I_{i-1}}$.
Since there are at most $k$ nonzero blocks,
$k\|\bh_{\I_i}\|_{\infty,\I} \leq \|\bh_{\I_{i-1}}\|_\mixed$.
Substituting (\ref{eq:l2n}) into (\ref{eq:hcn}),
\begin{equation}
\label{eq:sum} \|\bh_{(\I_0 \cup \I_1)^c}\|_2 \leq
 k^{-1/2}
\sum_{i=1}^{\ell-2}\|\bh_{\I_{i}}\|_\mixed \leq  k^{-1/2}
\sum_{i=1}^{\ell-1}\|\bh_{\I_{i}}\|_\mixed=k^{-1/2}
\|\bh_{\I_0^c}\|_\mixed,
\end{equation}
where the equality is a result of the fact that
$\|\bc_1+\bc_2\|_\mixed=\|\bc_1\|_\mixed+\|\bc_2\|_\mixed$ if
$\bc_1$ and $\bc_2$ are non-zero on disjoint blocks.

To develop a bound on  $\|\bh_{\I_0^c}\|_\mixed$ note that  since
$\bc'$ is a solution to (\ref{eq:l1}), $\|\bc_0\|_\mixed \geq
\|\bc'\|_\mixed$. Using the fact that
$\bc'=\bc_0+\bh_{\I_0}+\bh_{\I_0^c}$ and $\bc_0$ is supported on
$\I_0$ we have
\begin{equation}
\label{eq:c2i} \|\bc_0\|_\mixed \geq
\|\bc_0+\bh_{\I_0}\|_\mixed+\|\bh_{\I_0^c}\|_\mixed \geq
\|\bc_0\|_\mixed-\|\bh_{\I_0}\|_\mixed+\|\bh_{\I_0^c}\|_\mixed,
\end{equation}
from which we conclude that
\begin{equation}
\label{eq:hineq}  \|\bh_{\I_0^c}\|_\mixed \leq
\|\bh_{\I_0}\|_\mixed \leq k^{1/2}\|\bh_{\I_0}\|_2.
\end{equation}
The last inequality follows from applying Cauchy-Schwarz to any
block $k$-sparse vector $\bc$:
\begin{equation}
\label{eq:cs} \|\bc\|_\mixed =\sum_{|i|=k} \|\bc[i]\|_2 \cdot 1
\leq k^{1/2}\|\bc\|_2.
\end{equation}
Substituting (\ref{eq:hineq}) into (\ref{eq:sum}):
\begin{equation}
\label{eq:h2c} \|\bh_{(\I_0 \cup \I_1)^c}\|_2 \leq
\|\bh_{\I_0}\|_2 \leq  \|\bh_{\I_0 \cup \I_1}\|_2,
\end{equation}
which completes the first part of the proof.

{\it Part II:$\|\bh_{\I_0 \cup \I_1}\|_2 =0$}

We next show that $\bh_{\I_0 \cup \I_1}$ must be equal to $\bo$.
In this part we invoke the RIP.

 Since $\bbd\bc_0=\bbd\bc'=\by$, we have $\bbd\bh=\bo$. Using
 the fact that $\bh=\bh_{\I_0 \cup \I_1}+\sum_{i \geq 2} \bh_{\I_i}
 $,
 \begin{equation}
 \label{eq:isum}
\|\bbd\bh_{\I_0 \cup
\I_1}\|_2^2=-\sum_{i=2}^{\ell-1}\inner{\bbd(\bh_{\I_0}+\bh_{\I_1})}{\bbd\bh_{\I_i}}.
 \end{equation}
From the parallelogram identity and the block-RIP it can be shown
that
\begin{equation}
\label{eq:p3} |\inner{\bbd \bc_1}{\bbd\bc_2}| \leq
\delta_{2k}\|\bc_1\|_2\|\bc_2\|_2,
 \end{equation}
for any
 two block $k$-sparse vectors with disjoint support. The proof is
similar to \cite[Lemma 2.1]{C08} for the standard RIP.
% \begin{equation}
%\label{eq:p1} |\inner{\bbd \bc_1}{\bbd\bc_2}|=\frac{1}{4}\left|
%\|\bbd(\bc_1+\bc_2)\|_2^2-\|\bbd(\bc_1-\bc_2)\|_2^2\right|.
% \end{equation}
%If $\bc_1$ and $\bc_2$ are block $k$-sparse with different support
%and unit norm, then $\bc_1 \pm \bc_2$ is block $2k$-sparse and
%$\|\bc_1 \pm \bc_2\|_2^2=\|\bc_1\|_2^2+\|bc_2\|_2^2=2$. From
%(\ref{eq:rip}) it then follows that
% \begin{equation}
%\label{eq:p2} \left|
%\|\bbd(\bc_1+\bc_2)\|_2^2-\|\bbd(\bc_1-\bc_2)\|_2^2\right| \leq 2
%\delta_{2k}( \|\bc_1\|_2^2+\|\bc_2\|_2^2)=4\delta_{2k}.
% \end{equation}
Therefore,
 \begin{equation}
\left|\inner{\bbd\bh_{\I_0}}{\bbd\bh_{\I_i}}\right| \leq
\delta_{2k} \|\bh_{\I_0}\|_2\|\bh_{\I_i}\|_2,
 \end{equation}
 and similarly for $\inner{\bbd\bh_{\I_1}}{\bbd\bh_{\I_i}}$.
Substituting into (\ref{eq:isum}),
 \begin{eqnarray}
 \label{eq:isum2}
\|\bbd\bh_{\I_0 \cup \I_1}\|_2^2 & = &
\left|\sum_{i=2}^{\ell-1}\inner{\bbd(\bh_{\I_0}+\bh_{\I_1})}{\bbd\bh_{\I_i}}
\right| \nonumber \\
&\leq & \sum_{i=2}^{\ell-1}
\left(|\inner{\bbd\bh_{\I_0}}{\bbd\bh_{\I_i}}|+|\inner{\bbd\bh_{\I_1}}{\bbd\bh_{\I_i}}|\right)
\nonumber \\
& \leq & \delta_{2k}
(\|\bh_{\I_0}\|_2+\|\bh_{\I_1}\|_2)\sum_{i=2}^{\ell-1}
\|\bh_{\I_i}\|_2.
 \end{eqnarray}
From the Cauchy-Schwarz inequality, any length-2 vector $\ba$
satisfies $\ba(1)+\ba(2) \leq \sqrt{2}\|\ba\|_2$. Therefore,
 \begin{equation}
\|\bh_{\I_0}\|_2+\|\bh_{\I_1}\|_2 \leq\sqrt{2}
\sqrt{\|\bh_{\I_0}\|_2^2+\|\bh_{\I_1}\|_2^2}= \sqrt{2} \|\bh_{\I_0
\cup \I_1}\|_2,
 \end{equation}
 where the last equality is a result of the fact that $\bh_{\I_0}$
 and $\bh_{\I_1}$ have disjoint support. Substituting into
 (\ref{eq:isum2}) and using (\ref{eq:l2n}), (\ref{eq:sum}) and (\ref{eq:hineq}),
\begin{eqnarray}
\label{eq:ineqtmp} \|\bbd \bh_{\I_0 \cup \I_1}\|_2^2 &
\stackrel{(\ref{eq:l2n}), (\ref{eq:sum})}{\leq} &
\sqrt{2}k^{-1/2}\delta_{2k} \|\bh_{\I_0 \cup \I_1}\|_2
\|\bh_{\I_0^c}\|_{2,\I}
\nonumber \\
& \stackrel{(\ref{eq:hineq})}{\leq} & \sqrt{2}
\delta_{2k}\|\bh_{\I_0 \cup \I_1}\|_2 \|\bh_{\I_0}\|_2
\nonumber \\
 & \leq & \sqrt{2} \delta_{2k}\|\bh_{\I_0 \cup
\I_1}\|_2^2,
 \end{eqnarray}
where  the last inequality follows from $\|\bh_{\I_0}\|_2 \leq
\|\bh_{\I_0 \cup
 \I_1}\|_2$.
 Combining (\ref{eq:ineqtmp}) with the RIP (\ref{eq:rip}) we have
 \begin{equation}
 \label{eq:ripc}
(1-\delta_{2k})\|\bh_{\I_0 \cup \I_1}\|_2^2\leq  \|\bbd \bh_{\I_0
\cup \I_1}\|_2^2 \leq \sqrt{2}\delta_{2k} \|\bh_{\I_0 \cup
\I_1}\|_2^2.
 \end{equation}
Since $\delta_{2k}<\sqrt{2}-1$, (\ref{eq:ripc}) can hold only if
$\|\bh_{\I_0 \cup \I_1}\|_2=0$, which completes the proof.
\end{proof}

We conclude this subsection by pointing out more explicitly the
differences between the proof of Theorem~\ref{thm:rip} and that of
\cite{C08}. The main difference begins in (\ref{eq:l2n}); in our
formulation each of the subvectors $\bh_{\I_i}$ may have a
different number of non-zero elements, while the equivalent
equation in
 \cite{C08} (Eq. (10)) relies on the fact that the maximal number of non-zero elements in each of the subvectors
 is the same. This requires the use of several mixed-norms in our
setting. The rest of the proof follows the spirit of \cite{C08}
 where in some of the inequalities conventional norms are used, while in others
 the adaptation to our setting necessitates mixed norms.

\subsection{Robust Recovery}

We now treat the situation in which the observations are noisy,
and the vector $\bc_0$ is not exactly block-$k$ sparse.

Specifically, suppose that the measurements (\ref{eq:samplesd})
are corrupted by bounded noise so that
\begin{equation}
\by=\bbd\bc+\bz,
\end{equation}
where $\|\bz\|_2 \leq \epsilon$. In order to recover $\bc$ we use
the modified SOCP:
\begin{eqnarray}
\label{eq:l1n}
\min_{\bc} && \|\bc\|_{2,\I} \nonumber \\
\st && \|\by-\bbd\bc\|_2 \leq \epsilon.
\end{eqnarray}
In addition, given a $\bc \in \RR^N$, we denote by $\bc^k$ the
best approximation of $\bc$ by a vector with $k$ non-zero blocks,
so that $\bc^k$ minimizes
 $\|\bc-\bd\|_\mixed$ over all
block $k$-sparse vectors $\bd$.
 Theorem~\ref{thm:noise} shows that even when $\bc$ is not block $k$-sparse and
the measurements are noisy, the best block-$k$ approximation can
be well approximated using (\ref{eq:l1n}).
\begin{theorem}
\label{thm:noise} Let $\by=\bbd\bc_0+\bz$ be noisy measurements of
a vector $\bc_0$. Let $\bc^k$ denote the best block $k$-sparse
approximation of $\bc_0$, such that $\bc^k$ is block $k$-sparse
and minimizes $\|\bc_0-\bd\|_\mixed$ over all block $k$-sparse
vectors $\bd$, and let $\bc'$ be a solution to (\ref{eq:l1n}).
 If $\bbd$ satisfies the block RIP (\ref{eq:rip}) with
$\delta_{2k}<\sqrt{2}-1$ then
\begin{equation}
\label{eq:thmb} \|\bc_0-\bc'\|_2 \leq
\frac{2(1-\delta_{2k})}{1-(1+\sqrt{2})\delta_{2k}} k^{-1/2}
\|\bc_0-\bc^k\|_\mixed+\frac{4\sqrt{1+\delta_{2k}}}{1-(1+\sqrt{2})\delta_{2k}}\epsilon.
\end{equation}
\end{theorem}
Before proving the theorem, note that the first term in
(\ref{eq:thmb}) is a result of the fact that $\bc_0$ is not
exactly $k$-block sparse.  The second expression quantifies the
recovery error due to the noise.

\begin{proof}
The proof is very similar to that of Theorem~\ref{thm:rip} with a
few differences which we indicate. These changes follow the proof
of \cite[Theorem 1.3]{C08}, with appropriate modifications to
address the mixed norm.

 Denote by
$\bc'=\bc_0+\bh$ the solution to (\ref{eq:l1n}). Due to the noise
and the fact that $\bc_0$ is not block $k$-sparse, we will no
longer obtain $\bh=\bo$. However, we will show that $\|\bh\|_2$ is
bounded. To this end, we begin as in the proof of
Theorem~\ref{thm:rip} by using (\ref{eq:normh}). In the first part
of the proof we show that $\|\bh_{(\I_0 \cup \I_1)^c}\|_2 \leq
\|\bh_{\I_0 \cup \I_1}\|_2+2 e_0$ where
$e_0=k^{-1/2}\|\bc_0-\bc_{\I_0}\|_\mixed$ and $\bc_{\I_0}$ is the
restriction of $\bc_0$ onto the $k$ blocks corresponding to the
largest $\ell_2$ norm. Note that $\bc_{\I_0}=\bc^k$. In the second
part, we develop a bound on $\|\bh_{\I_0 \cup \I_1}\|_2$.

{\it Part I: Bound on $\|\bh_{(\I_0 \cup \I_1)^c}\|_2$}

We begin by decomposing $\bh$ as in the proof of
Theorem~\ref{thm:rip}. The inequalities until (\ref{eq:c2i}) hold
here as well. Instead of (\ref{eq:c2i}) we have
\begin{equation}
\label{eq:c2in} \|\bc_0\|_\mixed \geq
\|\bc_{\I_0}+\bh_{\I_0}\|_\mixed+\|\bc_{\I_0^c}+\bh_{\I_0^c}\|_\mixed
\geq
\|\bc_{\I_0}\|_\mixed-\|\bh_{\I_0}\|_\mixed+\|\bh_{\I_0^c}\|_\mixed-\|\bc_{\I_0^c}\|_\mixed.
\end{equation}
Therefore,
\begin{equation}
\label{eq:hineq2} \|\bh_{\I_0^c}\|_\mixed \leq
2\|\bc_{\I_0^c}\|_\mixed+\|\bh_{\I_0}\|_\mixed,
\end{equation}
where we used the fact that $\|\bc_0\|_\mixed-
\|\bc_{\I_0}\|_\mixed=\|\bc_{\I_0^c}\|_\mixed$. Combining
(\ref{eq:sum}), (\ref{eq:cs}) and (\ref{eq:hineq2})  we have
\begin{equation}
\label{eq:h2ct} \|\bh_{(\I_0 \cup \I_1)^c}\|_2 \leq
\|\bh_{\I_0}\|_2+2e_0 \leq  \|\bh_{\I_0 \cup \I_1}\|_2+2e_0,
\end{equation}
where $e_0=k^{-1/2}\|\bc_0-\bc_{\I_0}\|_\mixed$.

{\it Part II: Bound on $\|\bh_{\I_0 \cup \I_1}\|_2$}

 Using
 the fact that $\bh=\bh_{\I_0 \cup \I_1}+\sum_{i \geq 2} \bh_{\I_i} $ we have
 \begin{equation}
 \label{eq:isum3}
\|\bbd\bh_{\I_0 \cup \I_1}\|_2^2=\inner{\bbd\bh_{\I_0 \cup
\I_1}}{\bbd\bh}-\sum_{i=2}^{\ell-1}\inner{\bbd(\bh_{\I_0}+\bh_{\I_1})}{\bbd\bh_{\I_i}}.
 \end{equation}
From (\ref{eq:rip}),
\begin{equation}
 \label{eq:isum4}
|\inner{\bbd\bh_{\I_0 \cup \I_1}}{\bbd\bh}| \leq \|\bbd\bh_{\I_0
\cup \I_1}\|_2 \|\bbd\bh\|_2 \leq \sqrt{1+\delta_{2k}} \|\bh_{\I_0
\cup \I_1}\|_2 \|\bbd\bh\|_2.
 \end{equation}
Since both $\bc'$ and $\bc_0$ are feasible
\begin{equation}
\label{eq:err} \|\bbd\bh\|_2=\|\bbd(\bc_0-\bc')\|_2 \leq
\|\bbd\bc_0-\by\|_2+\|\bbd\bc'-\by\|_2 \leq 2\epsilon,
\end{equation}
and (\ref{eq:isum4}) becomes
\begin{equation}
 \label{eq:isum5}
|\inner{\bbd\bh_{\I_0 \cup \I_1}}{\bbd\bh}|\leq
2\epsilon\sqrt{1+\delta_{2k}} \|\bh_{\I_0 \cup \I_1}\|_2.
 \end{equation}
Substituting into (\ref{eq:isum3}),
\begin{eqnarray}\label{eq:cmb}
\|\bbd \bh_{\I_0 \cup \I_1}\|_2^2 & \leq &
\left|\inner{\bbd\bh_{\I_0 \cup
\I_1}}{\bbd\bh}\right|+\sum_{i=2}^{\ell-1}\left|\inner{\bbd(\bh_{\I_0}+\bh_{\I_1})}{\bbd\bh_{\I_i}}\right|
\nonumber \\
& \leq  & 2\epsilon \sqrt{1+\delta_{2k}}\|\bh_{\I_0 \cup
\I_1}\|_2+\sum_{i=2}^{\ell-1}\left|\inner{\bbd(\bh_{\I_0}+\bh_{\I_1})}{\bbd\bh_{\I_i}}\right|.
 \end{eqnarray}
Combining with (\ref{eq:isum2}) and (\ref{eq:ineqtmp}),
\begin{equation}
\|\bbd \bh_{\I_0 \cup \I_1}\|_2^2 \leq \bl 2\epsilon
\sqrt{1+\delta_{2k}}+\sqrt{2}\delta_{2k}
k^{-1/2}\|\bh_{\I_0^c}\|_\mixed \br \|\bh_{\I_0 \cup \I_1}\|_2.
 \end{equation}
Using (\ref{eq:cs}) and (\ref{eq:hineq2}) we have the upper bound
\begin{equation}
\label{eq:ub} \|\bbd \bh_{\I_0 \cup \I_1}\|_2^2 \leq \bl 2\epsilon
\sqrt{1+\delta_{2k}}+\sqrt{2}\delta_{2k}(\|\bh_{\I_0}\|+2e_0) \br
\|\bh_{\I_0 \cup \I_1}\|_2.
 \end{equation}

On the other hand, the RIP results in the lower bound
 \begin{equation}
\label{eq:lb2} \|\bbd \bh_{\I_0 \cup \I_1}\|_2^2 \geq
(1-\delta_{2k})\|\bh_{\I_0 \cup \I_1}\|_2^2.
 \end{equation}
From (\ref{eq:ub}) and (\ref{eq:lb2}),
 \begin{equation}
\label{eq:ineq} (1-\delta_{2k})\|\bh_{\I_0 \cup \I_1}\|_2 \leq
2\epsilon \sqrt{1+\delta_{2k}}+\sqrt{2}\delta_{2k} (\|\bh_{\I_0
\cup \I_1}\|+2e_0),
 \end{equation}
 or
 \begin{equation}
\label{eq:bound} \|\bh_{\I_0 \cup \I_1}\|_2 \leq
\frac{2\sqrt{1+\delta_{2k}}}{1-(1+\sqrt{2})\delta_{2k}}\epsilon
+\frac{2\sqrt{2}\delta_{2k}}{1-(1+\sqrt{2})\delta_{2k}}e_0.
 \end{equation}
The condition $\delta_{2k}<\sqrt{2}-1$ ensures that the
denominator in (\ref{eq:bound}) is positive. Substituting
(\ref{eq:bound}) results in
\begin{equation}
\label{eq:normhf} \|\bh\|_2  \leq \|\bh_{\I_0 \cup
\I_1}\|_2+\|\bh_{(\I_0 \cup \I_1)^c}\|_2 \leq 2\|\bh_{\I_0 \cup
\I_1}\|_2+2e_0,
\end{equation}
which completes the proof of the theorem.
\end{proof}

To summarize this section we have seen that as long as $\bbd$
satisfies the block-RIP (\ref{eq:rip}) with a suitable constant,
any block-$k$ sparse vector can be perfectly recovered from its
samples $\by=\bbd\bc$ using the convex SOCP (\ref{eq:l1}). This
algorithm is stable in the sense that by slightly modifying it as
in (\ref{eq:l1n}) it can tolerate noise in a way that ensures that
the norm of the recovery error is bounded by the noise level.
Furthermore, if $\bc$ is not block $k$-sparse, then its best
block-sparse approximation can be approached by solving the SOCP.
These results are summarized in Table~\ref{table:results}. In the
table, $\delta_{2k}$ refers to the block RIP constant.
\begin{table}[h] \caption{Comparison of algorithms for signal
recovery from $\by=\bbd\bc_0+\bz$} \begin{center}
%\begin{tabular}{|p{1.3in}|p{1.25in}|p{1.25in}|}
\begin{tabular}{|c|c|c|}
\hline &  Algorithm (\ref{eq:l1})  &  Algorithm
(\ref{eq:l1n}) \\ \hline $\bc_0$ & block $k$-sparse & arbitrary
\\ \hline Noise $\bz$ & none ($\bz=\bo$) & bounded $\|\bz\|_2 \leq
\epsilon$
 \\ \hline
 Condition on $\bbd$ & $\delta_{2k} \leq \sqrt{2}-1$ & $\delta_{2k} \leq
 \sqrt{2}-1$\\ \hline
 Recovery $\bc'$ &  $\bc'=\bc_0$ & $\|\bc_0-\bc'\|_2$ small;
 see (\ref{eq:thmb}) \\
\hline \end{tabular} \end{center} \label{table:results}
\end{table}

\subsection{Advantage of Block Sparsity}

The standard sparsity model considered in CS assumes that $\bx$
has at most $k$ non-zero elements, however it does not impose any
further structure. In particular, the non-zero components can
appear anywhere in the vector. There are many practical scenarios
in which the non-zero values are aligned to blocks, meaning they
appear in regions, and are not arbitrarily spread throughout the
vector. One example in the structured union of subspaces model we
treat in this paper. Other examples are considered in
\cite{PVMH08}.

Prior work on recovery of block-sparse vectors \cite{SPH08}
assumed consecutive blocks of the same size. It was sown that in
this case, when $n,N$ go to infinity, the algorithm (\ref{eq:l1})
will recover the true block-sparse vector with overwhelming
probability. Their analysis is based on characterization of the
null space of $\bbd$. In contrast, our approach relies on RIP
which allows the derivation of uniqueness and equivalence
conditions for finite dimensions and not only in the asymptotic
regime. In addition, Theorem~\ref{thm:noise} considers the case of
mismodelling and noisy observations while in \cite{SPH08} only the
ideal noise-free setting is treated.

To demonstrate the advantage of our algorithm over standard basis
pursuit (\ref{eq:l1cs}), consider the matrix $\bbD$ of
(\ref{matD}). In Section~\ref{sec:l1}, the standard and block RIP
constants of $\bbD$ were calculated and it was shown that block
RIP constants are smaller. This suggests that there are input
vectors $\bx$ for which the mixed $\ell_2/\ell_1$ method of
(\ref{eq:l1}) will be able to recover them exactly from
measurements $\by=\bbd\bc$ while standard $\ell_1$ minimization
will fail. To illustrate this behavior, let
$\bx=[0,0,1,-1,-1,0.1]^T$ be a $4$-sparse vector, in which the
non-zero elements are known to appear in blocks of length $2$. The
prior knowledge that $\bx$ is 4-sparse is not sufficient to
determine $\bx$ from $\by$. In contrast, there is a unique
block-sparse vector consistent with $\by$. Furthermore, our
algorithm which is a relaxed version of (\ref{eq:l0}), finds the
correct $\bx$ while standard $\ell_1$ minimization fails in this case; its output is $\hat{\bx}=[-0.0289,0,0.9134,-1.0289,-1.0289,0]$.

We further compare the recovery performance of $\ell_1$
minimization (\ref{eq:l1cs}) and our algorithm (\ref{mainalg}) for
an extensive set of random signals. In the experiment, we draw a
matrix $\bbD$ of size $25 \times 50$ from the Gaussian ensemble.
The input vector $\bx$ is also randomly generated as a
block-sparse vector with blocks of length $5$. We draw $1\leq k
\leq 25$ non-zero entries from a zero-mean unit variance normal
distribution and divide them into blocks which are chosen
uniformly at random within $\bx$. Each of the algorithms is
executed based on the measurements $\by=\bbd\bx$. In
Fig.~\ref{FigCompareBP_l2l1} we plot the fraction of successful
reconstructions for each $k$ over $500$ experiments. The results
illustrate the advantage of incorporating the block-sparsity
structure into the optimization program. An interesting feature of
the graph is that when using the block-sparse recovery approach,
the performance is roughly constant over the block-length ($5$ in
this example). This explains the performance advantage over
standard sparse recovery.
\begin{figure}
\centering
\includegraphics[scale=1]{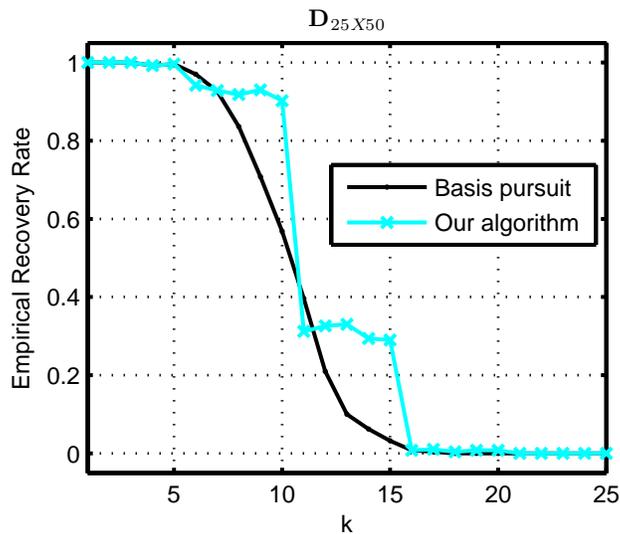}
\caption{Recovery rate of block-sparse signals using standard
$\ell_1$ minimization (basis pursuit) and the mixed
$\ell_2/\ell_1$ algorithm.}\label{FigCompareBP_l2l1}
\end{figure}

\section{Application to MMV Models}\label{sec:mmv}

We now specialize our algorithm and equivalence results to the MMV
problem. This leads to two contributions which we discuss in this
section: The first is an equivalence result based on RIP for a
mixed-norm MMV algorithm. The second is a new measurement strategy
in MMV problems that leads to improved performance over
conventional MMV methods, both in simulations and as measured by
the RIP-based equivalence condition. In contrast to previous
equivalence results, for this strategy we show that even if we
choose the worst possible $\bbx$, improved performance over the
single measurement setting can be guaranteed.

\subsection{Equivalence Results}

As we have seen in Section~\ref{sec:union}, a special case of
block sparsity is the MMV model, in which we are given a
 matrix of measurements $\bby=\bbm \bbx$ where $\bbx$ is an unknown $L
\times d$ matrix that has at most $k$ non-zero rows. Denoting by
$\bc=\mvec(\bbx^T), \by=\mvec(\bby^T)$, $\bbd=\bbm^T \otimes
\bbi_d$ we can express the vector of measurements $\by$ as
$\by=\bbd\bc$ where $\bc$ is a block sparse vector with
consecutive blocks of length $d$. Therefore, the results of
Theorems~\ref{thm:rip} and \ref{thm:noise} can be specified to
this problem.

Recovery algorithms for MMV using convex optimization programs
were studied in \cite{Chen,TroppII} and  several greedy algorithms
were proposed in \cite{Cotter,TroppI}. Specifically, in
\cite{Chen,Cotter,TroppI,TroppII} the authors study a class of
optimization programs, which we refer to as M-BP:
\begin{equation}\label{mmvopt}
\textrm{M-BP($\ell_q$):}\qquad\min\sum_{i=1}^L \|\bbx^i\|_q^p\quad
\st \bby=\bbm \bbx,
\end{equation}
where $\bbx^i$ is the $i$th row of $\bbx$. The choice
$p=1,q=\infty$ was considered in \cite{TroppII}, while \cite{Chen}
treated the case of $p=1$ and arbitrary $q$. Using $p\leq 1$ and
$q=2$ was suggested in \cite{Cotter},\cite{Malioutov}, leading to
the iterative algorithm M-FOCUSS. For $p=1,q=2$, the program
(\ref{mmvopt}) has a global minimum which M-FOCUSS is proven to
find. A nice comparison between these methods can be found in
\cite{TroppII}. Equivalence for MMV algorithms based on RIP
analysis does not appear in previous papers. The most detailed
theoretical analysis can be found in \cite{Chen} which establishes
equivalence results based on mutual coherence.  The results imply
equivalence for (\ref{mmvopt}) with $p=1$ under conditions equal
to those obtained for the single measurement case. Note that RIP
analysis typically leads to tighter equivalence bounds than mutual
coherence analysis.

In our recent work \cite{ME08}, we suggested an alternative
approach to solving MMV problems by merging the $d$ measurement
columns with random coefficients and in such a way transforming
the multiple measurement problem into a single measurement
counterpart. As proved in \cite{ME08}, this technique preserves
the non-zero location set with probability one thus reducing
computational complexity. Moreover, we showed that this method can
be used to boost the empirical recovery rate by repeating the
random merging several times.

Using the block-sparsity approach we can alternatively cast any
MMV model as a single measurement vector problem by
deterministically transforming the multiple measurement vectors
into the single vector model $\mvec(\bby^T)=(\bbm \otimes
\bbi_d)\mvec(\bbx^T)$, where $\bc= \mvec(\bbx^T)$ is block-$k$
sparse with consecutive blocks of length $d$. In contrast to
\cite{ME08} this does not reduce the number of unknowns so that
the computational complexity of the resulting algorithm is on the
same order as previous approaches, and also does not offer the
opportunity for boosting. However, as we see in the next
subsection, with an appropriate choice of measurement matrix this
approach results in improved recovery capabilities.

Since we can cast the MMV problem as one of block-sparse recovery,
we may apply our equivalence results of Theorem~\ref{thm:rip} to
this setting leading to RIP-based equivalence. To this end we
first note that applying the SOCP (\ref{eq:l1}) to the effective
measurement vector $\by$ is the same as  solving (\ref{mmvopt})
with $p=1,q=2$. Thus the equivalence conditions we develop below
relate to this program. Next, if $\bz=\bbd\bc$ where $\bc$ is a
block $2k$-sparse vector and $\bbd=\bbm \otimes \bbi_d$, then
taking the structure of $\bbd$ into account, $\bbz=\bbm \bbx$
where $\bbx$ is a size $L \times d$ matrix whose $i$th row is
equal to $\bc[i]$, and similarly for $\bbz$. The block sparsity of
$\bc$ implies that $\bbx$ has at most $2k$ non-zero rows. The
squared $\ell_2$ norm $\|\bz\|_2^2$ is equal to the squared
$\ell_2$ norm of the rows of $\bbz$ which can be written as
\begin{equation}
\|\bz\|_2^2=\|\bbz\|_F^2=\tr(\bbz^T\bbz),
\end{equation}
where $\|\bbz\|_F$ denotes the Frobenius norm. Since
$\|\bc\|_2^2=\|\bbx\|_F^2$ the RIP  condition becomes
\begin{equation}
\label{eq:mmvrip1} (1-\delta_{2k}) \tr(\bbx^T\bbx) \leq
\tr(\bbx^T\bbm^T\bbm\bbx) \leq (1+\delta_{2k}) \tr(\bbx^T\bbx),
\end{equation}
for any $L \times d$ matrix $\bbx$  with at most $2k$ non-zero
rows.

We now show that (\ref{eq:mmvrip1}) is equivalent to the standard
RIP condition
\begin{equation}
\label{eq:mmvrip2} (1-\delta_{2k}) \|\bx\|_2^2 \leq
\|\bbm\bx\|_2^2 \leq (1+\delta_{2k}) \|\bx\|_2^2,
\end{equation}
for any length $L$ vector $\bx$ that is $2k$-sparse. To see this,
suppose first that (\ref{eq:mmvrip1}) is satisfied for every
matrix $\bbx$ with at most $2k$ non-zero rows and let $\bx$ be an
arbitrary $2k$-sparse vector. If we define $\bbx$ to be the matrix
whose columns are all equal to $\bx$, then $\bbx$ will have at
most $2k$ non-zero rows and therefore satisfies
(\ref{eq:mmvrip1}). Since the columns of $\bbx$ are all equal,
$\tr(\bbx^T\bbx)=d\|\bx\|_2^2$ and
$\tr(\bbx^T\bbm^T\bbm\bbx)=d\|\bbm\bx\|_2^2$ so that
(\ref{eq:mmvrip2}) holds. Conversely, suppose that
(\ref{eq:mmvrip2}) is satisfied for all $2k$-sparse vectors $\bx$
and let $\bbx$ be an arbitrary matrix with at most $2k$ non-zero
rows. Denoting by $\bx_j$ the columns of $\bbx$, each $\bx_j$ is
$2k$-sparse and therefore satisfies (\ref{eq:mmvrip2}). Summing
over all values $j$ results in (\ref{eq:mmvrip1}).

To summarize, if $\bbm$ satisfies the conventional RIP condition
(\ref{eq:mmvrip2}), then the algorithm (\ref{mmvopt}) with
$p=1,q=2$ will recover the true unknown $\bbx$. This requirement
reduces to that we would obtain if we tried to recover each column
of $\bbx$ separately, using the standard $\ell_1$ approach
(\ref{eq:l1cs}). As we already noted, previous equivalence results
for MMV algorithms also share this feature. Although this
condition guarantees that processing the vectors jointly does not
harm the recovery ability, in practice exploiting the joint
sparsity pattern of $\bbx$ via (\ref{mmvopt}) leads to improved
results. Unfortunately, this behavior is not captured by any of
the known equivalence conditions. This is due to the special
structure of $\bbd=\bbm \otimes\bbi$. Since each measurement
vector $\by_i$ is affected only by the corresponding vector
$\bx_i$, it is clear that in the worst-case we can choose
$\bx_i=\bx$ for some vector $\bx$. In this case, all the $\by_i$s
are equal so that adding measurement vectors will not improve our
recovery ability. Consequently, worst-case analysis based on the
standard measurement model for MMV problems cannot lead to
improved performance over the single measurement case.

\subsection{Improved MMV Recovery}

We have seen that the pessimistic equivalence results for MMV
algorithms is a consequence of the fact that in the worst-case
scenario in which $\bx_i=\bx$, using a separable measurement
strategy will render all observation vectors equal. In this
subsection we introduce an alternative measurement technique for
MMV problems that can lead to improved worst-case behavior, as
measured by RIP, over the single channel case.

One way to improve the analytical results is to consider an
average case analysis instead of a worst-case approach.
 In
\cite{ER09} we show that if the unknown vectors $\bx_i$ are
generated randomly, then the performance improves with increasing
number of measurement vectors. The advantage stems from the fact
that the situation of equal vectors has zero probability and
therefore does not affect the average performance. Here we take a
different route which does not involve randomness in the unknown
vectors, and leads to improved results even in the worst-case
(namely without requiring an average analysis).

To enhance the performance of MMV recovery, we note that when we
allow for an arbitrary (unstructured) $\bbd$, the RIP condition of
Theorem~\ref{thm:rip} is weaker than the standard RIP requirement
for recovering $k$-sparse vectors. This suggests that we can
improve the performance of MMV methods by converting the problem
into a general block sparsity problem, and then sampling with an
arbitrary unstructured matrix $\bbd$ rather than the choice
$\bbd=\bbm^T \otimes \bbi_d$. The tradeoff introduced is increased
computational complexity since each measurement is based on all
input vectors. The theoretical conditions will now be looser,
since block-RIP is weaker than standard RIP. Furthermore, in
practice, this approach often improves the performance over
separable MMV measurement techniques as we illustrate in the
following example.

In the example, we compare the performance of several MMV
algorithms for recovering $\bbx$ in the model $\bby=\bbm\bbx$,
with our method based on block sparsity in which the measurements
$\by$ are obtained via $\by=\bbd\bc$ where $\bc=\mvec(\bbx^T)$ and
$\bbd$ is a dense matrix. Choosing $\bbd$ as a block diagonal
matrix with blocks equal to $\bbm$ results in the standard MMV
measurement model. The effective matrices $\bbd$ have the same
size in the case in which it is block diagonal and when it is
dense. To compare  the performance of  (\ref{mainalg}) with a
dense $\bbd$ to that of (\ref{mmvopt}) with a block diagonal
$\bbd$, we compute the empirical recovery rate of the methods in
the same way performed in \cite{ME08}. The matrices $\bbm$ and
$\bbd$ are drawn randomly from a Gaussian ensemble. In our
example, we choose $\ell=20,L=30,d=5$ where $\ell$ is the number
of rows in $\bby$. The matrix $\bbx$ is generated randomly by
first selecting the $k$ non-zero rows uniformly at random, and
then drawing the elements in these rows from a normal
distribution. The empirical recovery rates using the methods of
(\ref{mmvopt}) for different choices of $q$ and $p$, ReMBO
\cite{ME08} and our algorithm (\ref{eq:l1}) with dense $\bbd$ are
depicted in Fig.~\ref{FigMMV_BlockRIP}. When the index $p$ is
omitted it is equal to $1$. Evidently, our algorithm performs
better than most popular optimization techniques for MMV systems.
\begin{figure}
\centering
\includegraphics[scale=1]{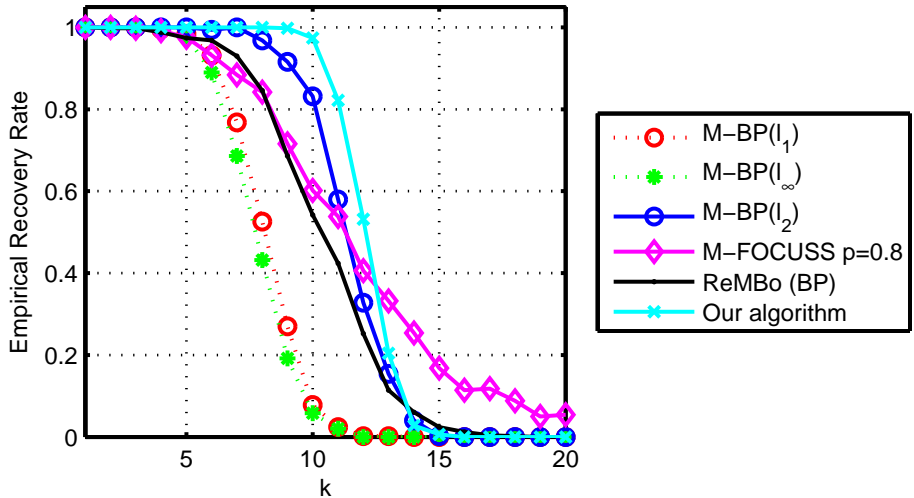}
\caption{Recovery rate for different number $k$ of non-zero rows
in $\bbx$. Each point on the graph represents an average recovery
rate over 500 simulations.}\label{FigMMV_BlockRIP}
\end{figure}
We stress that the performance advantage is due to the joint
measurement process rather than a new recovery algorithm.

\section{Random Matrices}\label{sec:prob}

Theorems~\ref{thm:rip} and \ref{thm:noise} establish that a
sufficiently small block RIP constant $\delta_{2k|\I}$ ensures
exact recovery of the coefficient vector $\bc$. We now prove that
random matrices are likely to satisfy this requirement.
Specifically, we show that the probability that $\delta_{k|\I}$
exceeds a certain threshold decays exponentially in the length of
$\bc$. Our approach relies on results of \cite{CT05},\cite{CT06}
developed for standard RIP, however, exploiting the block
structure of $\bc$ leads to a much faster decay rate.

\begin{proposition}\label{prop:randomD}
Suppose $\bbD$ is an $n\times N$ matrix from the Gaussian
ensemble, namely $[\bbD]_{ik}\sim\mathcal{N}(0,\frac{1}{n})$. Let
$\delta_{k|\I}$ be the smallest value satisfying the block RIP
(\ref{eq:rip}) over $\I=\{d_1=d,\ldots,d_m=d\}$, assuming $N=md$
for some integer $m$. Then, for every $\epsilon>0$ the block RIP
constant $\delta_{k|\I}$ obeys (for $n,N$ large enough, and fixed
$d$)
\begin{equation}\label{eq:randomD}
\prob\left(\sqrt{1+\delta_{k|\I}}>1+(1+\epsilon)f(r)\right)\leq
2\expp{-NH(r)\epsilon}\cdot\expp{-m(d-1)H(r)}.
\end{equation}
Here, the ratio $r=kd/N$ is fixed,
$f(r)=\sqrt{\frac{N}{n}}\left(\sqrt{r}+\sqrt{2H(r)}\right)$, and
$H(q)=-q\log q -(1-q) \log(1-q)$ is the entropy function defined
for $0<q<1$.
\end{proposition}

The assumption that $d_i=d$ simplifies the calculations in the
proof. Following the proof, we shortly address the more difficult
case in which the blocks have varying lengths. We note that
Proposition~\ref{prop:randomD} reduces to the result of
\cite{CT05} when $d=1$. However, since $f(r)$ is independent of
$d$, it follows that for $d>1$ and fixed problem dimensions
$n,N,r$, block-RIP constants are smaller than the standard RIP
constant. The second exponent in the right-hand side of
(\ref{eq:randomD}) is responsible for this behavior.

\begin{proof}
Let $\lambda=(1+\epsilon)f(r)$ and define
\begin{equation}\label{eq:rndsig}
\bar{\sigma} = \max_{|T|=k,d}\sigmamax (\bbD_T),\quad
\underline{\sigma} = \min_{|T|=k,d}\sigmamin(\bbD_T),
\end{equation}
where $\sigmamax(\bbD_T),\sigmamin(\bbD_T),$ are the largest and
the smallest singular values of $\bbD_T$, respectively. We use
$|T|=k,d$ to denote a column subset of $\bbD$ consisting of $k$
blocks of length $d$. For brevity we omit subscripts and denote
$\delta=\delta_{k|\I}$. The inequalities in the definition of
block-RIP (\ref{eq:rip}) imply that
\begin{eqnarray}
1+\delta \geq \bar{\sigma}^2\\
1-\delta \leq \underline{\sigma}^2.
\end{eqnarray}
Since $\delta$ is the smallest number satisfying these
inequalities we have that $1+\delta =
\max(\bar{\sigma}^2,2-\underline{\sigma}^2)$. Therefore,
\begin{align}\label{eq:proof_random1a}
\prob\left(\sqrt{1+\delta}>1+\lambda\right) & =
\prob\left(\sqrt{\max(\bar{\sigma}^2,2-\underline{\sigma}^2)}>1+\lambda\right)\\
& \leq
\prob(\bar{\sigma}>1+\lambda)+\prob(\sqrt{2-\underline{\sigma}^2}>1+\lambda).
\end{align}
Noting that $\underline{\sigma}\geq 1-\lambda$ implies
$\sqrt{2-\underline{\sigma}^2} \leq 1+\lambda$ we conclude that
\begin{equation}\label{eq:proof_random1}
\prob\left(\sqrt{1+\delta}>1+\lambda\right)\leq
\prob(\bar{\sigma}>1+\lambda)+\prob(\underline{\sigma}<1-\lambda).
\end{equation}

We now bound each term in the right-hand-side of
(\ref{eq:proof_random1}) using a result of Davidson and Szarek
\cite{SZ91} regarding the concentration of the extreme singular
values of a Gaussian matrix. It was proved in \cite{SZ91} that an
$m\times n$ matrix $\bbX$ with $n\geq m$ satisfies
\begin{eqnarray}
\prob(\sigmamax(\bbX)>1+\sqrt{m/n}+t)\leq \expp{-nt^2/2}\\
\prob(\sigmamin(\bbX)<1-\sqrt{m/n}-t)\leq \expp{-nt^2/2}.
\end{eqnarray}
Applying a union bound leads to
\begin{align}\label{eq:unbd}
\prob\left(\bar{\sigma}>1+\sqrt{\frac{kd}{n}}+t\right) & \leq
\sum_{|T|=k,d}\prob\left(\sigmamax(\bbD_T)>1+\sqrt{\frac{kd}{n}}+t\right)\\&\leq
\sum_{|T|=k,d}\expp{-nt^2/2}\\&= {m \choose k}\expp{-nt^2/2}.
\end{align}
Using the well-known bound on the binomial coefficient (for
sufficiently large $m$)
\begin{equation}
{m \choose k}\leq\expp{mH(k/m)},
\end{equation}
we conclude that
\begin{equation}\label{eq:proof_random2}
\prob\left(\bar{\sigma}>1+\sqrt{\frac{kd}{n}}+t\right)\leq
\expp{mH(k/m)}\expp{-nt^2/2}.
\end{equation}

To utilize this result in (\ref{eq:proof_random1}) we rearrange
\begin{align}
1+\lambda & = 1+(1+\epsilon)f(r)\label{eq:proof_random2a}\\
& = 1+(1+\epsilon)\left(\sqrt{\frac{kd}{n}}+\sqrt{\frac{2N}{n}H(r)}\right)\\
& \geq
1+\sqrt{\frac{kd}{n}}+\sqrt{(1+\epsilon)\frac{2N}{n}H(r)}\label{eq:proof_random2b}
\end{align}
and obtain that
\begin{align}
\prob\left(\bar{\sigma}>1+\lambda\right)&\leq
\prob\left(\bar{\sigma} >
1+\sqrt{\frac{kd}{n}}+\sqrt{(1+\epsilon)\frac{2N}{n}H(r)}\right).
\end{align}
Using (\ref{eq:proof_random2}) leads to
\begin{align}
\prob\left(\bar{\sigma}>1+\lambda\right) & \leq \expp{mH(k/m)}\expp{-\frac{n(1+\epsilon) 2NH(r)}{2n}}\\
& = \expp{NH(r)-m(d-1)H(r)-(1+\epsilon) NH(r)}\\
&\leq
\expp{-NH(r)\epsilon}\expp{-m(d-1)H(r)}\label{eq:proof_random3}.
\end{align}
Similar arguments are used to bound the second term in
(\ref{eq:proof_random1}), completing the proof.
\end{proof}

The proof of Proposition~\ref{prop:randomD} can be adapted to the
case in which $d_i$ are not equal. In this case, the notation
$|T|=k,d$ is replaced by $|T|=k|\I$ and has the following meaning:
$T$ indicates a column subset of $\bbD$ consisting of $k$ blocks
from $\I$. Since $\I$ contains variable-length blocks, $|T|$ is
not constant and depends on the particular column subset.
Consequently, in order to apply the union bounds in
(\ref{eq:unbd}) we need to consider the worst-case scenario
corresponding to the maximal block length in $\I$.
Proposition~\ref{prop:randomD} thus holds for $d=\max(d_i)$.
However, it is clear that the resulting probability bound will not
be as stringent as in the case of equal $d_i=d$, especially when
the ratio $\max(d_i)/\min(d_i)$ is large.

Proposition~\ref{prop:randomD} holds as is for matrices $\bbD$
from the Bernoulli ensemble, namely $[\bbD]_{ik}=\pm
\frac{1}{\sqrt{n}}$ with equal probability. In fact, the
proposition is true for any ensemble for which the concentration
of extreme singular values holds.

The following corollary emphasizes the asymptotic behavior of
block-RIP constants per given number of samples.

\begin{corollary}
Consider the setting of Proposition~\ref{prop:randomD}, and define
$g(r)=\sqrt{\frac{N}{n}}\left(\sqrt{r}+\sqrt{2H(r)d^{-1}}\right)$.
Then,
\begin{equation}\label{eq:randomD2}
\prob\left(\sqrt{1+\delta_{k|\I}}>1+(1+\epsilon)g(r)\right)\leq
2\expp{-mH(r)\epsilon}.
\end{equation}
\end{corollary}

\begin{proof}
Let $\lambda=(1+\epsilon)g(r)$. The result then follows by
replacing (\ref{eq:proof_random2a})-(\ref{eq:proof_random2b}) with
\begin{equation}
1+\lambda \geq
1+\sqrt{\frac{kd}{n}}+\sqrt{(1+\epsilon)\frac{2N}{nd}H(r)},
\end{equation}
which leads to $\prob(\bar{\sigma}>1+\lambda) \leq
\expp{-mH(r)\epsilon}$.
\end{proof}

To evaluate the asymptotic behavior of block-RIP we note that for
every $\epsilon>0$ the right-hand side of (\ref{eq:randomD2}) goes
to zero when $N=md\rightarrow \infty$. Consequently,  for fixed
$d$
\begin{equation}\label{eq:num_meas}
\delta_{k|\I} < \rho(r) \deft -1+[1+g(r)]^2,
\end{equation}
with overwhelming probability. In Fig.~\ref{FigBlockRIP} we
compute $\rho(r)$ for several problem dimensions and compare it
with standard RIP which is obtained when $d=1$. Evidently, as the
non-zero entries are forced to block structure, a wider range of
sparsity ratios $r$ satisfy the condition of
Theorem~\ref{thm:rip}.

\begin{figure}
\centering
\includegraphics[scale=1]{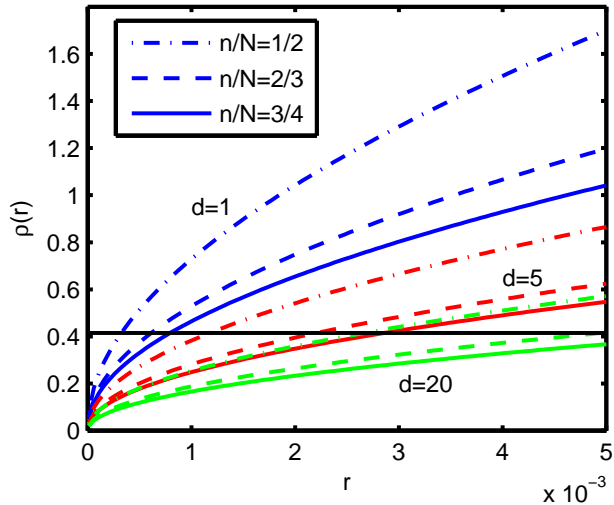}
\caption{The upper bound on $\delta_{k|\I}$ as a function of the
sparsity ratio $r$, for three sampling rates $n/N$, and three
block structures $d=1,5,20$. The horizontal threshold is fixed on
$\rho^\ast=\sqrt{2}-1$ representing the threshold for equivalence
derived in Theorem~\ref{thm:rip}.}\label{FigBlockRIP}
\end{figure}

Although Fig.~\ref{FigBlockRIP} shows advantage for block-RIP, the
absolute sparsity ratios predicted by the theory are pessimistic
as also noted in \cite{CT05},\cite{CT06} in the case of $d=1$. To
offer a more optimistic viewpoint, the RIP and block-RIP constants
were computed brute-force for several instances of $\bbD$ from the
Gaussian ensemble. Fig.~\ref{FigEmpricialBlockRip} plots the
results and qualitatively affirms that block-RIP constants are
more ``likely" to be smaller than their standard RIP counterparts,
even when the dimensions $n,N$ are relatively small.

\begin{figure*}
\centering \mbox {
\subfigure[]{\includegraphics[scale=0.62]{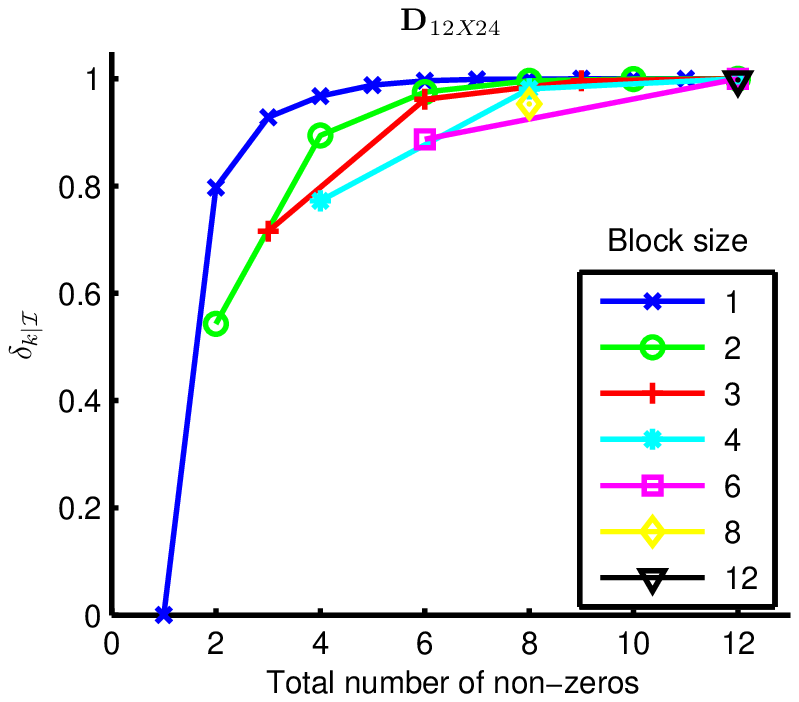}}
\subfigure[]{\includegraphics[scale=0.62]{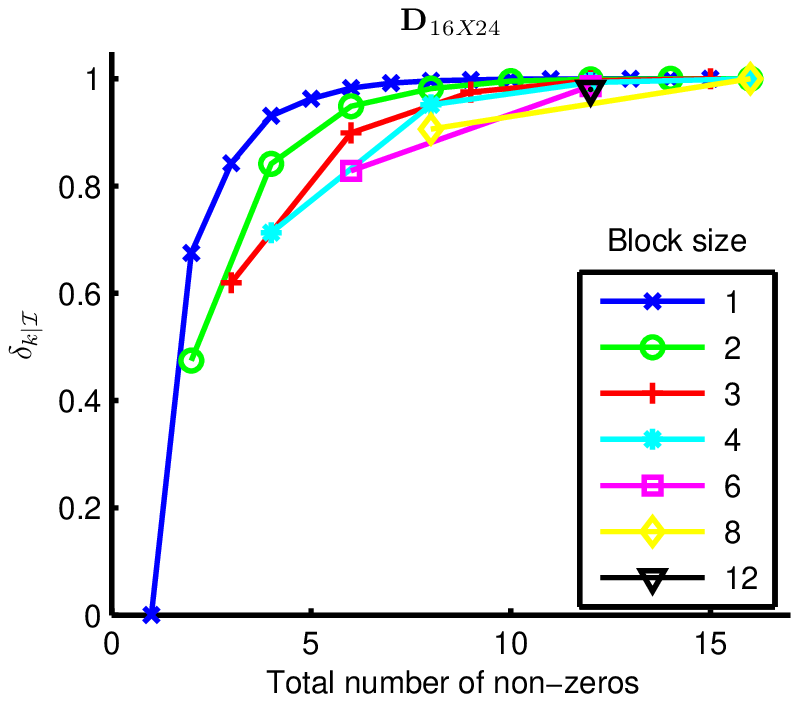}}
\subfigure[]{\includegraphics[scale=0.62]{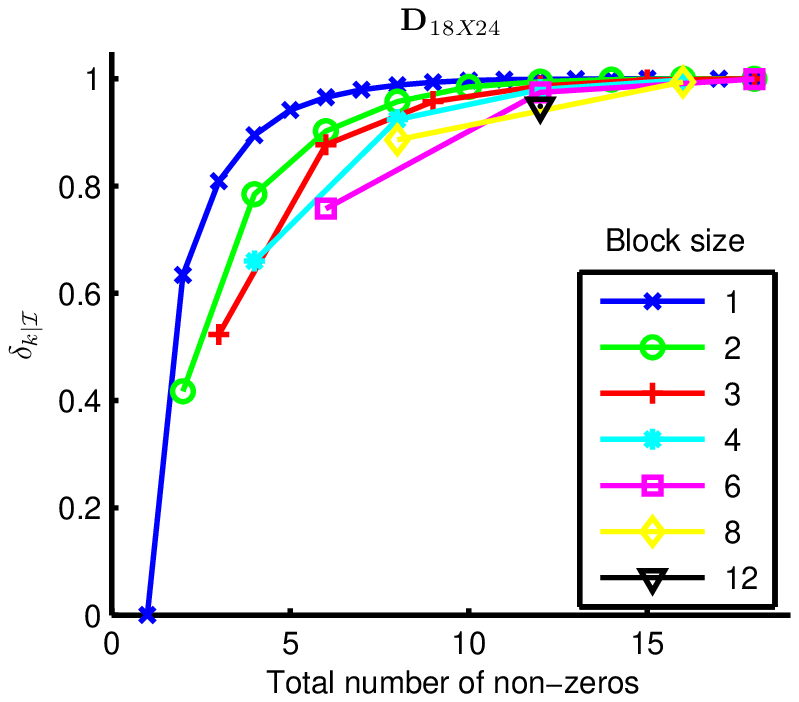}}}
\caption{The standard and block-RIP constants $\delta_{k|\I}$ for
three different dimensions $n,N$. Each graph represent an average
over 10 instances of random matrix $\bbD$. Each instance of $\bbD$
is scaled by a factor such that (\ref{eq:stablity}) is satisfied
with $\alpha+\beta=2$.} \label{FigEmpricialBlockRip}
\end{figure*}

An important question is how many samples are needed roughly in
order to guarantee stable recovery. This question is addressed in
the following proposition, which quotes a result from \cite{BD08}
based on the proofs of \cite{BDDW08}; we rephrase the result to
match our notation.
\begin{proposition}[{\cite[Theorem 3.3]{BD08}}]\label{Corr:Thomas}
Consider the setting of Proposition~\ref{prop:randomD}, namely a
random Gaussian matrix $\bbD$ of size $n\times N$ and block sparse
signals over $\I=\{d_1=d,\ldots,d_m=d\}$, where $N=md$ for some
integer $m$. Let $t>0$ and $0 < \delta < 1$ be constant numbers.
If
\begin{equation}\label{eq:nummeasrandom}
n \geq
\frac{36}{7\delta}\left(\ln(2L)+kd\ln\bl\frac{12}{\delta}\br+t\right),
\end{equation}
where $L=\binom{m}{k}$, then $\bbD$ satisfies the block-RIP
(\ref{eq:rip}) with restricted isometry constant
$\delta_{k|\I}=\delta$, with probability at least $1-e^{-t}$.
\end{proposition}

As observed in \cite{BD08}, the first term in
(\ref{eq:nummeasrandom}) has the dominant impact on the required
number of measurements in an asymptotic sense. Specifically, for
block sparse signals
\begin{equation}
(m/k)^{k} \leq L={m \choose k} \leq (e\,m/k)^{k}.
\end{equation}
Thus, for a given fraction of nonzeros $r=kd/N$, roughly $n\approx
k\log(m/k)=-k\log(r)$ measurements are needed. For comparison, to
satisfy the standard RIP a larger number $n\approx -kd\log(r)$ is
required. Note that Corollary~\ref{Corr:Thomas} puts the emphasis
on the required problem dimensions to satisfy a given RIP level.
In contrast, Proposition~\ref{prop:randomD} provides a tail bound
on the expected isometry constant for given problem dimensions.

\section{Conclusion}

In this paper, we studied the problem of recovering an unknown
signal $x$ in an arbitrary Hilbert space $\HH$, from a given set
of $n$ samples which are modelled as inner products of $x$ with
sampling functions $s_i,1 \leq i \leq n$. The signal $x$ is known
to lie in a union of subspaces, so that $x \in \V_i$ where each of
the subspaces $\V_i$ is a sum of $k$ subspaces $\A_i$ chosen from
an ensemble of $m$ possibilities. Thus, there are ${m \choose k}$
possible subspaces in which $x$ can lie, and a-priori we do not
know which subspace is the true one. While previous treatments of
this model considered invertibility conditions, here we provide
concrete recovery algorithms for a signal over a structured union.

We began by showing that recovering $x$ can be reduced to a
sparsity problem in which the goal is to recover a block-sparse
vector $\bc$ from measurements $\by=\bbd \bc$ where the non-zero
values in $\bc$ are grouped into blocks. The measurement matrix
$\bbd$ is equal to $S^*A$ where $S^*$ is the sampling operator and
$A$ is  a set transformation corresponding to a basis for the sum
of all $\A_i$. To determine $\bc$ we suggested a mixed
$\ell_2/\ell_1$ convex optimization program that takes on the form
of an SOCP. Relying on the notion of block-RIP, we developed
sufficient conditions under which $\bc$ can be perfectly recovered
using the proposed algorithm.  We also proved that under the same
conditions, the unknown $\bc$ can be stably approximated in the
presence of noise. Furthermore, if $\bc$ is not exactly
block-sparse, then its best block-sparse approximation can be
approached using the proposed method.
 We then showed that when $\bbd$ is
chosen at random, the recovery conditions are satisfied with high
probability.

Specializing the results to MMV systems, we proposed a new method
for sampling in MMV problems. In this approach each measurement
vector depends on all the unknown vectors. As we showed, this can
lead to better recovery rate. Furthermore, we established
equivalence results for a class of MMV algorithms based on RIP.

Throughout the paper, we assumed a finite union of subspaces as
well as finite dimension of the underlying spaces. An interesting
future direction to explore is the extension of the ideas
developed herein to the more challenging problem of recovering $x$
in a possibly infinite union of subspaces, which are not
necessarily finite-dimensional. Although at first sight this seems
like a difficult problem as our algorithms are inherently
finite-dimensional, recovery methods for sparse signals in
infinite dimensions have been addressed in some of our previous
work \cite{ME07,ME08,ME09,E08,E082}. In particular, we have shown
that a signal lying in a union of shift-invariant subspaces can be
recovered efficiently from certain sets of sampling functions. In
our future work, we intend to combine these results with those in
the current paper in order to develop a more general theory for
recovery from a union of subspaces.

A recent preprint \cite{MCS08} that was posted online after the
submission of this paper proposes a new framework called
model-based compressive sensing (MCS). The MCS approach assumes a
vector signal model in which only certain predefined sparsity
patterns may appear. In general, obtaining efficient recovery
algorithms in such scenarios is difficult, unless further
structure is imposed on the sparsity patterns. Therefore, the
authors consider two types of sparse vectors: block sparsity as
treated here, and a wavelet tree model. For these settings, they
generalize two known greedy algorithms: CoSaMP \cite{NT09} and
iterative hard thresholding (IHT) \cite{BD08}. These results
emphasize our claim that theoretical questions of uniqueness and
stable representation can be studied for arbitrary unions as in
\cite{BD09}. However tractable recovery algorithms inherently
require some structure, as the one considered here.

The union model developed in this paper is broader than the
block-sparse setting treated in \cite{MCS08} in the sense that it
allows to model linear dependencies between the nonzero values
rather than only between their locations, by appropriate choice of
subspaces in (\ref{eq:union}), (\ref{eq:unionv}). In addition, we
aim at optimization-based recovery algorithms
(\ref{eq:l1}),(\ref{eq:l1n}) which require selecting the objective
in order to promote the model properties. Finally, we emphasize
that our results are non asymptotic and also ensure stable
recovery in the presence of noise and signal mismodeling.

\begin{singlespace}
\bibliographystyle{IEEEtran}
\bibliography{IEEEabrv,notes,paper,newbib}
\end{singlespace}

\end{document}